\documentclass[11pt,notitlepage,tightenlines,nofootinbib,superscriptaddress,aps,pra]{revtex4-2}
\pdfoutput=1

\usepackage{newpxtext,newpxmath}

\usepackage[utf8]{inputenc}
\usepackage{amsthm}
\usepackage{amssymb}
\usepackage{amsmath}
\usepackage{bm,bbm}
\usepackage{braket}
\usepackage{mathtools}
\usepackage{enumerate}
\usepackage[shortlabels]{enumitem}
\usepackage{graphicx}
\usepackage{stackengine}
\usepackage{scalerel}
\usepackage{tensor}       %\tensor[_n]{\braket{j|\psi}}{_{1\ldots n}}
\usepackage[dvipsnames,svgnames]{xcolor}
\usepackage{makecell}
\newcolumntype{x}[1]{>{\centering\arraybackslash}p{#1}}
\usepackage{tikz}
\usepackage{pgfplots}
\usetikzlibrary{shapes.geometric, shapes.misc, positioning, arrows, arrows.meta, decorations.pathreplacing, decorations.pathmorphing, patterns, angles, quotes, calc}
\usepackage{booktabs}
\usepackage{xfrac}
\usepackage{centernot}
\usepackage{comment}
\usepackage{chngcntr}

\usepackage[pdftex]{hyperref}
\hypersetup{
    bookmarksnumbered=true, % If Acrobat bookmarks are requested, include section numbers
    breaklinks=true,
    unicode=false, % non-Latin characters in Acrobat bookmarks
    pdfstartview={FitH}, % fits the width of the page to the window
    pdfnewwindow=true, % links in new window
    colorlinks=true, % false: boxed links; true: colored links
    allcolors=MidnightBlue!70!black!70!TealBlue
}

\DeclareFontFamily{U}{mathx}{\hyphenchar\font45}
\DeclareFontShape{U}{mathx}{m}{n}{
      <5> <6> <7> <8> <9> <10>
      <10.95> <12> <14.4> <17.28> <20.74> <24.88>
      mathx10
      }{}
\DeclareSymbolFont{mathx}{U}{mathx}{m}{n}
\DeclareFontSubstitution{U}{mathx}{m}{n}
\DeclareMathAccent{\widecheck}{0}{mathx}{"71}

\newtheorem{thm}{Theorem}
\newtheorem*{thm*}{Theorem}

\newtheorem*{theorem*}{Theorem}

\newtheorem*{prop*}{Proposition}
\newtheorem{proposition}[thm]{Proposition}
\newtheorem*{proposition*}{Proposition}
\newtheorem{lemma}[thm]{Lemma}
\newtheorem*{lemma*}{Lemma}

\newtheorem{corollary}[thm]{Corollary}
\newtheorem*{cor*}{Corollary}

\newtheorem*{conjecture*}{Conjecture}

\newtheorem*{Def*}{Definition}

\newtheorem*{question*}{Question}

\newtheorem*{problem*}{Problem}

\DeclareMathAlphabet{\mathcal}{OMS}{cmsy}{m}{n}

\makeatletter
\def\thmhead@plain#1#2#3{%
  \thmname{#1}\thmnumber{\@ifnotempty{#1}{ }\@upn{#2}}%
  \thmnote{ {\the\thm@notefont#3}}}
\let\thmhead\thmhead@plain
\makeatother

\theoremstyle{definition}

\newtheorem*{rem*}{Remark}
\newtheorem*{remark}{Remark}

\newcommand{\bb}{\begin{equation}\begin{aligned}\hspace{0pt}}
\newcommand{\bbb}{\begin{equation*}\begin{aligned}}
\newcommand{\ee}{\end{aligned}\end{equation}}
\newcommand{\eee}{\end{aligned}\end{equation*}}
\newcommand\floor[1]{\left\lfloor#1\right\rfloor}
\newcommand\ceil[1]{\left\lceil#1\right\rceil}

\newcommand{\ketbraa}[2]{\ket{#1}\!\bra{#2}}

\renewcommand{\epsilon}{\varepsilon}

\newcommand{\ve}{\varepsilon}

\newcommand{\SEP}{\mathrm{SEP}}

\DeclareMathOperator{\Tr}{Tr}

\DeclareMathAlphabet{\pazocal}{OMS}{zplm}{m}{n}

\DeclareMathOperator{\supp}{supp}

\DeclareMathOperator{\tr}{tr}

\newcommand{\lsmatrix}{\left(\begin{smallmatrix}}
\newcommand{\rsmatrix}{\end{smallmatrix}\right)}

\stackMath

\stackMath

\makeatletter
\newcommand*\rel@kern[1]{\kern#1\dimexpr\macc@kerna}
\newcommand*\widebar[1]{%
  \begingroup
  \def\mathaccent##1##2{%
    \rel@kern{0.8}%
    \overline{\rel@kern{-0.8}\macc@nucleus\rel@kern{0.2}}%
    \rel@kern{-0.2}%
  }%
  \macc@depth\@ne
  \let\math@bgroup\@empty \let\math@egroup\macc@set@skewchar
  \mathsurround\z@ \frozen@everymath{\mathgroup\macc@group\relax}%
  \macc@set@skewchar\relax
  \let\mathaccentV\macc@nested@a
  \macc@nested@a\relax111{#1}%
  \endgroup
}

\counterwithin*{equation}{part}
\counterwithin*{thm}{part}
\counterwithin*{figure}{part}

\allowdisplaybreaks

\let\nc\newcommand
\renewcommand{\bar}{\;\rule{0pt}{9.5pt}\right|\;}
\nc{\lset}{\left\{\left.}
\nc{\rset}{\right\}}
\nc{\lsetr}{\left\{\,}
\nc{\rsetr}{\right.\right\}}
\nc{\barr}{\;\rule{0pt}{9.5pt}\left|\;}

\newcommand{\wt}{\widetilde}

\let\textgeq\relax
\let\texteq\relax

\newcommand{\texteq}[1]{\stackrel{\mathclap{\mbox{\scriptsize#1}}}{=}}

\newcommand{\textgeq}[1]{\stackrel{\mathclap{\mbox{\scriptsize#1}}}{\geq}}

\usepackage[most,breakable]{tcolorbox}
\renewenvironment{boxed}[1][white]%
  {\expandafter\ifstrequal\expandafter{#1}{filled}{%
    \begin{tcolorbox}[colback=MidnightBlue!70!black!70!TealBlue!2!white,colframe=MidnightBlue!70!black!70!TealBlue!30!white,breakable,enhanced,left=5.75pt,right=5.75pt,grow sidewards by=10pt]%
  }{%
  \expandafter\ifstrequal\expandafter{#1}{red}{%
    \begin{tcolorbox}[colback=red!5!white,colframe=red!80!black!40!white,breakable,enhanced,left=5.75pt,right=5.75pt,grow sidewards by=10pt]%
  }{%else
    \begin{tcolorbox}[colback=white,colframe=gray!15,breakable,enhanced,left=5.75pt,right=5.75pt,grow sidewards by=10pt]%
  }%
  }%
  }%
  {\end{tcolorbox}}

\nc{\R}{\mathcal{R}}
\nc{\W}{\mathcal{W}}
\nc{\T}{\mathcal{T}}
\nc{\B}{\mathcal{B}}
\nc{\C}{\mathcal{C}}
\nc{\U}{\mathcal{U}}
\nc{\E}{\mathcal{E}}
\nc{\EE}{\mathscr{E}}
\nc{\FF}{\mathscr{F}}
\nc{\K}{\mathcal{K}}
\nc{\V}{\mathcal{V}}
\nc{\X}{\mathcal{X}}
\nc{\F}{\mathcal{F}}
\nc{\G}{\mathcal{G}}
\nc{\D}{\mathcal{D}}
\nc{\Y}{\mathcal{Y}}

\nc{\M}{\mathcal{M}}
\nc{\N}{\mathcal{N}}
\nc{\I}{\mathcal{I}}

\nc{\OO}{\mathscr{O}}
\nc{\Q}{\mathbb{Q}}
\nc{\RR}{\mathbb{R}}
\nc{\CC}{\mathbb{C}}

\nc{\HH}{\mathbb{H}}
\nc{\MM}{\mathbb{M}}
\nc{\NN}{\mathbb{N}}
\nc{\DD}{\mathbb{D}}
\nc{\LL}{\mathcal{L}}

\nc{\id}{\mathbbm{1}}
\nc{\idc}{\mathrm{id}}

\nc{\norm}[2]{\left\lVert#1\right\rVert_{\,#2}}
\nc{\proj}[1]{\ket{#1}\!\bra{#1}}
\nc{\lnorm}[2]{\left\lVert#1\right\rVert_{\ell_{#2}}}

\let\oldproofname\proofname
\renewcommand{\proofname}{\rm\bf{\oldproofname}}

\makeatletter
\renewenvironment{proof}[1][\proofname]{\par
\pushQED{\qed}%
\normalfont \topsep6\p@\@plus6\p@\relax
\trivlist
\item\relax
{\bfseries  %<- font shape
#1\@addpunct{.}}\hspace\labelsep\ignorespaces %<- includes punctuation code
}{%
\popQED\endtrivlist\@endpefalse
}
\makeatother

\nc{\rhos}{\rho'}

\nc{\gm}{\mathbin\#}

\newcommand{\Rep}{\mathfrak{R}}

\newcommand{\dia}[1]{\norm{#1}{\!\Diamond}}

\newcommand{\sPnorm}{P,=}
\newcommand{\sPsub}{P,\leq}
\newcommand{\sTnorm}{T,=}
\newcommand{\sTsub}{T,\leq}
\newcommand{\sTPnorm}{T/P}
\newcommand{\sTPsub}{T/P,\leq}
\newcommand{\sMnorm}{M,=}
\newcommand{\sMsub}{M,\leq}
\newcommand{\sDnorm}{\Diamond}
\newcommand{\sMDnorm}{\wh{\Diamond}}
\newcommand{\smoothing}[2]{#2,\,#1}
\NewDocumentCommand{\s}{m m O{\ve}}{%
 \lowercase{\def\sdist{#1}}
 \lowercase{\def\snorm{#2}}
 \IfEqCase{\sdist}{%
  {t}{%
   \IfEqCase{\snorm}{%
    {norm}{ \smoothing{\sTnorm}{#3} }%
    {n}{ \smoothing{\sTnorm}{#3} }%
    {sub}{ \smoothing{\sTsub}{#3}}%
    {s}{ \smoothing{\sTsub}{#3}}%
   }%
  }%
  {p}{%
   \IfEqCase{\snorm}{%
    {norm}{ \smoothing{\sPnorm}{#3} }%
    {n}{ \smoothing{\sPnorm}{#3} }%
    {sub}{ \smoothing{\sPsub}{#3}}%
    {s}{ \smoothing{\sPsub}{#3}}%
   }%
  }%
  {tp}{%
   \IfEqCase{\snorm}{%
    {norm}{ \smoothing{\sTPnorm}{#3} }%
    {n}{ \smoothing{\sTPnorm}{#3} }%
    {sub}{ \smoothing{\sTPsub}{#3}}%
    {s}{ \smoothing{\sTPsub}{#3}}%
   }%
  }%
  {m}{%
   \IfEqCase{\snorm}{%
    {norm}{ \smoothing{\sMnorm}{#3} }%
    {n}{ \smoothing{\sMnorm}{#3} }%
    {sub}{ \smoothing{\sMsub}{#3}}%
    {s}{ \smoothing{\sMsub}{#3}}%
   }%
  }%
  {dia}{%
   \smoothing{\sDnorm}{#3}%
  }%
  {mdia}{%
   \smoothing{\sMDnorm}{#3}%
  }%
 }%
}%
\NewDocumentCommand{\Dmax}{m m O{\ve}}{%
 D_{\max}^{\s{#1}{#2}[#3]}
}%
\NewDocumentCommand{\Dhmax}{m m O{\ve}}{%
 \wh{D}_{\max}^{#3}
}%
\NewDocumentCommand{\Dcmax}{m m O{\ve}}{%
 \wc{D}_{\max}^{#3}
}%

\usepackage{xstring}

\makeatletter
\def\fnum@figure{\textbf{Figure}\nobreakspace\textbf{\thefigure}}
\def\fnum@table{\textbf{Table}\nobreakspace\textbf{\thetable}}
\renewcommand{\@caption@fignum@sep}{: }
\renewcommand{\p@subsection}{\thesection.}% Default is \thesection\,
\renewcommand{\p@subsubsection}{\thesection.\thesubsection.}% Default is \thesection\,\thesubsection\,
\makeatother
\usepackage{multirow}

\newcommand{\wh}{\widehat}
\newcommand{\wc}{\widecheck}

\begin{document}

\title{Quantum information processing under maximal distance measures}

\author{Bartosz Regula}
\email{bartosz.regula@gmail.com}
\affiliation{Mathematical Quantum Information RIKEN Hakubi Research Team, RIKEN Pioneering Research Institute (PRI) and RIKEN Center for Quantum Computing (RQC), Wako, Japan}

\author{Ludovico Lami}
\email{ludovico.lami@gmail.com}
\affiliation{Scuola Normale Superiore, Piazza dei Cavalieri 7, 56126 Pisa, Italy}

\author{Mario Berta}
\email{berta@physik.rwth-aachen.de}
\affiliation{Institute for Quantum Information, RWTH Aachen University, Aachen, Germany}

\author{Marco Tomamichel}
\email{marco.tomamichel@nus.edu.sg}
\affiliation{Department of Electrical and Computer Engineering, National University of Singapore, Singapore}
\affiliation{Centre for Quantum Technologies, National University of Singapore, Singapore}

 \begin{abstract}
The measures of distance frequently used in quantum information, such as trace distance or fidelity-based purified distance, are the smallest extensions of the corresponding distance measures in classical information theory. However, this way of extending classical distances is not unique, and a whole spectrum of other extensions can be defined.
Here we study the question of how quantum information processing tasks are affected by a change of the error metrics to the \emph{maximal} extensions of the classical distance measures, imposing more demanding error requirements on operational tasks.
We find that the performance of some tasks, such as channel coding or quantum resource distillation, is not affected whatsoever by this change.
However, for other tasks such as channel simulation and quantum resource dilution, we show that the rates are now given not by functions of the standard (Umegaki) quantum relative entropy, but ones based on the Belavkin--Staszewski relative entropy. 
%In particular, we show that under a choice of maximal distance, the thermodynamic work cost of a quantum state under Gibbs-preserving operations is given by its Belavkin--Staszewski free energy, while the communication cost of simulating a quantum channel assisted by non-signalling correlations can be expressed by a regularised expression that we show to equal the channel's Belavkin--Staszewski mutual information for specific examples.
Our results connect the Belavkin--Staszewski relative entropy with operational tasks and reveal that known reversibility results in quantum information, including the reverse Shannon theorem and the asymptotic reversibility of quantum resource transformations, no longer hold under maximal distance measures.
 \end{abstract}

\maketitle

%%%%%%%%%%%%%%%%%%%%%%%%%%%%%%%%%%%%%%%%%%%%%%%%%%%%%%%%%%%%%%%%%%%%%%%%%%%%%%%%%%%
%%%%%%%%%%%%%%%%%%%%%%%%%%%%%%%%%%%%%%%%%%%%%%%%%%%%%%%%%%%%%%%%%%%%%%%%%%%%%%%%%%%
%%%%%%%%%%%%%%%%%%%%%%%%%%%%%%%%%%%%%%%%%%%%%%%%%%%%%%%%%%%%%%%%%%%%%%%%%%%%%%%%%%%

\section{Introduction}

The rates of standard quantum information processing tasks are evaluated in the asymptotic i.i.d.\ limit where many copies of a quantum state or channel can be used. This includes problems such as the capacity of quantum channels, the error exponents of quantum hypothesis testing, or the rates of distillation and dilution of quantum resources such as thermodynamic work or quantum entanglement. It is in this limit that fundamental entropic quantities emerge as operationally meaningful --- the interpretation as asymptotic i.i.d.\ rates gives meaning to the mutual information of a quantum channel~\cite{bennett_2002} or the relative entropy between different quantum states~\cite{hiai_1991}.

An important aspect of the asymptotic characterisation of tasks is that, in their standard formulations, they allow for some error in the process. The error of a communication procedure, for example, need not be precisely zero, as long as it vanishes in the asymptotic i.i.d.\ limit. Any such allowance for error must fix the distance measure with respect to which the error is evaluated, as different error metrics are not necessarily equivalent. 
Standard definitions in classical information theory use the statistical (total variation) distance, which naturally connects with guessing probabilities.  
In quantum information, one most commonly encounters one of two distance measures: either the trace distance, which generalises the classical statistical distance, or the purified distance, which quantifies the trace distance between purifications and connects with the Uhlmann fidelity. 

However, quantum extensions of classical distances and divergences are not unique in general. In particular, although both the trace and the purified distance are connected with useful operational interpretations, they are only the \emph{minimal} extensions of their classical counterparts, meaning that they are obtained from the latter by optimising the distances between measurement outcomes. However, in general one can define an infinite family of valid extensions~\cite{matsumoto_2010-1,matsumoto_2018}. Choosing an extension larger than the minimal one would then impose a more stringent criterion on the error allowed in the processing of quantum states or channels. 
The starting point of our work is the question: can the rates of quantum information processing tasks change when another extension of the relevant distances is chosen? We address this by going to the opposite end of the spectrum and investigating the behaviour of tasks when one instead selects the \emph{maximal} quantum extensions of the classical distance measures that still obey the data processing inequality. Such maximal extensions are obtained by looking at classical distances between preparations rather than measurements~\cite{matsumoto_2010-1}.

Our investigation reveals two properties. On the one hand, for tasks that are similar in nature to  distillation problems, no change whatsoever occurs in their performance when the distance measures are changed in this way. This means that that the corresponding rates (and even one-shot bounds) remain the same for all possible quantum extensions of the classical distances. This applies to problems such as thermodynamic work extraction, entanglement distillation, and quantum channel capacity assisted by shared entanglement. On the other hand, however, the opposite tasks connected with resource dilution --- such as the simulation of quantum channels with noiseless ones --- can change significantly. In particular, we show that their rates are no longer governed by standard entropic measures based on the familiar (Umegaki) quantum relative entropy. Instead, the underlying entropic measure is changed to the Belavkin--Staszewski relative entropy~\cite{belavkin_1982}, which is the maximal extension of the classical relative entropy~\cite{matsumoto_2010}. While it has attracted interest as a technical tool and a general bound~\cite{hiai_1991,csiszar_2007,carlen_2019,fang_2021,katariya_2021,bluhm_2025}, the Belavkin--Staszewski relative entropy has thus far not found a direct operational meaning in quantum information processing. Here we establish that it exactly quantifies the asymptotic i.i.d.\ rates of operational problems under maximal distance measures: (1) the cost of simulating a quantum state using thermodynamic work bits under Gibbs-preserving operations equals the free energy expressed in terms of the Belavkin--Staszewski relative entropy; (2) the entanglement cost of any pure state is given by its Belavkin--Staszewski relative entropy of entanglement; (3) the communication cost of simulating quantum channels with the assistance of non-signalling correlations equals a regularised expression that can be shown to exactly equal the Belavkin--Staszewski mutual information for several channels of interest, which we exemplify with the examples of pure classical--quantum channels and the qubit amplitude damping channel. 
On the technical side, our results rely on an asymptotic equipartition property for a maximal extension of the smooth max-relative entropy that we introduce and characterise here.

A major consequence of our findings concerns the phenomenon of reversibility, which is a feature found in a plethora of quantum resource theories: under suitable allowances for the free operations used, the rates of distillation and dilution processes are inverses of each other, meaning that transformations between quantum states or channels can be realised in a way that is asymptotically perfectly reversible. Importantly, this then allows for the identification of a unique notion of entropic measure of resources in the form of the Umegaki relative entropy~\cite{horodecki_2003,brandao_2010,horodecki_2013-1}. This remarkable property, conceptually similar to thermodynamic reversibility, was first established for the conversion of pure entangled states~\cite{bennett_1996-2}, later generalised to the theory of communication with entanglement assistance~\cite{bennett_2014,berta_2011} (where it is known as the reverse Shannon theorem), and to general quantum resources under more relaxed assumptions~\cite{brandao_2010,berta_2023,hayashi_2025-3,lami_2025-1}. Our results show that such phenomena are made impossible under maximal distance measures: all of the aforementioned types of reversibility are no longer valid in general, since the rates of distillation and dilution are now governed by fundamentally different measures of entropy. This shows that, although the change to maximal distance measures does not alter the asymptotic behaviour of all tasks, it nevertheless reshapes the landscape of quantum information processing --- the balance between distillation- and dilution-type problems is upended, with the two now governed by two different entropic quantities.

\begin{table*}[t]
\centering
\small
\renewcommand{\arraystretch}{1.5}
\setlength{\tabcolsep}{5pt}
\begin{tabular*}{\textwidth}{@{\extracolsep{\fill}}llcc@{}}
\toprule
\textbf{Setting} & \textbf{Task}
& \makecell{\textbf{Rate under}\\\textbf{standard distance measures}\\{}}
& \makecell{\textbf{Rate under}\\\textbf{maximal distance measures}\\\textcolor{blue!80!black!50}{[this work]}} \\
\midrule
\midrule
\multirow{2}{*}{Thermodynamics}
& Work extraction
& $D(\rho\|\gamma)$ & $D(\rho\|\gamma)$ \\
& Work dilution
& $D(\rho\|\gamma)$ & $\wh{D}(\rho\|\gamma)$ \\
\midrule
\multirow{2}{*}{\makecell[l]{Pure-state\\entanglement}}
& Entanglement distillation
& $D(\psi_{AB} \| \mathrm{SEP}) = S(\psi_A)$ & $D(\psi_{AB} \| \mathrm{SEP}) = S(\psi_A)$ \\
& Entanglement dilution
& $D(\psi_{AB} \| \mathrm{SEP}) = S(\psi_A)$ & $\wh{D}(\psi_{AB} \| \mathrm{SEP}) = S_{1/2}(\psi_A)$ \\
\midrule
\multirow{2}{*}{\makecell[l]{Assisted quantum\\communication}}
& Channel coding
& $\frac12 I(\N)$ & $\frac12 I(\N)$ \\
& Channel simulation
& $\frac12 I(\N)$
& \hphantom{\;\textsuperscript{$(*)$}}$\frac12\wh{I}(\N)$\;\textsuperscript{$[*]$} \\
\bottomrule
\end{tabular*}
\caption{\textbf{Comparison of asymptotic i.i.d.\ rates with vanishing error.} Thermodynamics is considered under
Gibbs-preserving operations, pure-state entanglement under non-entangling operations, and quantum communication under
no-signalling assistance. Here $\gamma$ is the Gibbs state,
$\psi_A=\Tr_B\psi_{AB}$ for a pure state $\psi_{AB}$,
$S$ is the von Neumann entropy, and
$S_{1/2}(\psi_A)=2\log\Tr\sqrt{\psi_A}$.
Work is measured in units of \smallskip 
$k_{\rm B}T\ln2$, entanglement in ebits, and communication in qubits. \\
$[*]$~The maximal distance channel simulation formula is established only for 
 example channel families; the general expression for all channels is given in terms of a regularisation of a smooth max-information. See Section~\ref{sec:comm} for details.}
\label{tab:main_results}
\end{table*}

%%%%%%%%%%%%%%%%%%%%%%%%%%%%%%%%%%%%%%%%%%%%%%%%%%%%%%%%%%%%%%%%%%%%%%%%%%%%%%%%%%%
%%%%%%%%%%%%%%%%%%%%%%%%%%%%%%%%%%%%%%%%%%%%%%%%%%%%%%%%%%%%%%%%%%%%%%%%%%%%%%%%%%%
%%%%%%%%%%%%%%%%%%%%%%%%%%%%%%%%%%%%%%%%%%%%%%%%%%%%%%%%%%%%%%%%%%%%%%%%%%%%%%%%%%%

\section{Preliminaries}

\subsection{Divergences and distances}

Let $\DD(\rho\|\sigma)$ denote any divergence between quantum states, that is, a quantity that obeys the data-processing inequality $\DD(\N(\rho)\|\N(\sigma)) \leq \DD(\rho\|\sigma)$ for any quantum channel $\N$. There is a systematic way to obtain such divergences from classical ones~\cite{matsumoto_2018}: given a divergence $\DD(p\|q)$ defined for classical distributions, its \emph{minimal} (measured) and \emph{maximal} extensions to quantum states are defined as
\begin{equation}\begin{aligned}
  \wc{\DD}(\rho\|\sigma) &\coloneqq \sup \lset \DD(p\|q) \bar p = \M(\rho), \; q = \M(\sigma),\; \M \in \mathrm{CPTP} \rset,\\
  \wh{\DD}(\rho\|\sigma) &\coloneqq \inf \lset \DD(p\|q) \bar \N(p) = \rho,\; \N(q) = \sigma,\; \N \in \mathrm{CPTP} \rset,
\end{aligned}\end{equation}
where the optimisation is over all classical distributions $p$ and $q$ and all quantum channels between the relevant spaces. Specifically, the minimal extension $\wc{\DD}$ optimises over all quantum-to-classical channels $\M$, which correspond to measurements; the maximal extension $\wh{\DD}$, on the other hand, optimises over all classical-to-quantum preparation channels $\N$. 
The terms `minimal' and `maximal' here refer to the fact that, for any quantum divergence $\DD'$ such that $\DD'(p\|q) = \DD(p\|q)$ and such that $\DD'$ satisfies data processing, we necessarily have $\wc{\DD}(\rho\|\sigma) \leq \DD'(\rho\|\sigma) \leq \wh{\DD}(\rho\|\sigma)$.

Two commonly used distance measures encountered in the study of quantum states are precisely the minimal extensions of the corresponding classical quantities: these are the \emph{trace distance}
\begin{equation}\begin{aligned}
    T(\rho,\sigma) = \wc{T}(\rho,\sigma) = \frac12 \norm{\rho - \sigma}{1}
\end{aligned}\end{equation}
and the \emph{purified distance}
\begin{equation}\begin{aligned}
    P(\rho,\sigma) = \wc{P}(\rho,\sigma) = \sqrt{1 - F(\rho,\sigma)}, \qquad F(\rho,\sigma) = \norm{\sqrt\rho\sqrt\sigma}{1}^2.
\end{aligned}\end{equation}
The Fuchs--van de Graaf inequalities~\cite{fuchs_1999} imply that the two distances can be mutually related as
\begin{equation}\begin{aligned}\label{eq:fvdg}
   T(\rho, \sigma) \leq P(\rho,\sigma) \leq \sqrt{2 T(\rho,\sigma)}.
\end{aligned}\end{equation}

Another important class of divergences are the R\'enyi divergences. For classical distributions, they are defined as
\begin{equation}\begin{aligned}
  D_\alpha(p\|q) = \frac{1}{\alpha-1} \log \sum_x p(x)^{\alpha} \, q(x)^{1-\alpha}.
\end{aligned}\end{equation}
Here, $\alpha \in [0,\infty]$, with the cases of $\alpha \in \{0, 1, \infty\}$ understood as the respective limits in $\alpha$. In the quantum case, several types of extensions are used. 
The sandwiched R\'enyi relative entropies~\cite{muller-lennert_2013,wilde_2014} are given by\footnote{For simplicity, in this section we restrict ourselves to expressions valid for $\sigma > 0$. General formulas for quantum R\'enyi divergences with a possibly rank-deficient $\sigma$ can be obtained by taking a limit $\ve\to0^+$ of the divergences evaluated at $\sigma + \ve \id$; we refer e.g.\ to~\cite{hiai_2017,khatri_2020} for details.}
\begin{equation}\begin{aligned}
    \wt{D}_\alpha (\rho \| \sigma) = \frac{1}{\alpha-1} \log \Tr\! \left[ \left(\sigma^{\frac{1-\alpha}{2\alpha}} \rho \sigma^{\frac{1-\alpha}{2\alpha}}\right)^\alpha \right],
\end{aligned}\end{equation}
and in the limit $\alpha \to \infty$ they yield the \emph{max-relative entropy}~\cite{datta_2009}:
\begin{equation}\begin{aligned}
    D_{\max}(\rho\|\sigma) = \lim_{\alpha\to\infty} \wt{D}_\alpha(\rho\|\sigma) = \log \inf \lset \lambda \bar \rho \leq \lambda \sigma \rset.
\end{aligned}\end{equation}
Another commonly encountered extension are the Petz--R\'enyi relative entropies $\widebar{D}_\alpha(\rho\|\sigma) = \frac{1}{\alpha-1} \log \Tr [ \rho^\alpha \sigma^{1-\alpha} ]$~\cite{petz_1986}. Both the sandwiched and the Petz R\'enyi divergences find operational use in the large-deviation analysis of many quantum information tasks. Interestingly, however, they are neither the minimal nor the maximal extensions of the classical R\'enyi divergences. 
The sandwiched R\'enyi relative entropies can in fact be expressed as a regularisation of the minimal (measured) R\'enyi divergences: for all $\alpha \geq \frac12$,~\cite{mosonyi_2015,hayashi_2016-1} 
\begin{equation}\begin{aligned}\label{eq:sandwiched_is_reg_measured}
    \wt{D}_\alpha (\rho \| \sigma) = \lim_{n\to\infty} \frac1n \wc{D}_\alpha\!\left(\rho^{\otimes n} \middle\| \sigma^{\otimes n}\right).
\end{aligned}\end{equation}
The maximal R\'enyi divergences, on the other hand, have a particularly convenient expression in the range $\alpha \in [0,2]$: there they equal the \emph{geometric R\'enyi divergences}~\cite{matsumoto_2018},
\begin{equation}\begin{aligned}
    \wh{D}_\alpha(\rho\|\sigma) = \frac{1}{\alpha-1} \log \Tr \left[ \sigma \left( \sigma^{-1/2} \rho \sigma^{-1/2} \right)^\alpha \right], \qquad \alpha \in [0,2].
\end{aligned}\end{equation}
Very recently, also the range of $\alpha > 2$ of the maximal R\'enyi extensions was characterised, albeit only after a regularisation~\cite{rubboli_2026-1}.

The limit $\alpha \to 1$ of the R\'enyi divergences is of particular importance. Classically, it corresponds to the conventional relative entropy $D_1(p\|q) = \sum_x p(x) \left( \log p(x)  - \log q(x)\right)$, often called the Kullback--Leibler (KL) divergence. In the quantum case, the sandwiched and Petz R\'enyi divergences recover in the limit $\alpha \to 1$ the standard (Umegaki) quantum relative entropy~\cite{umegaki_1962},
\begin{equation}\begin{aligned}
    \lim_{\alpha\to1} \wt{D}_\alpha(\rho\|\sigma) = D(\rho\|\sigma) = \Tr \rho \left(\log \rho - \log \sigma\right).
\end{aligned}\end{equation}
Similarly to Eq.~\eqref{eq:sandwiched_is_reg_measured}, it can be regarded as a regularisation of the minimal extension of the classical relative entropy. The limit $\alpha \to 1$ of the geometric R\'enyi divergences gives the maximal extension of the relative entropy, studied first by Belavkin and Staszewski~\cite{belavkin_1982} and given by
\begin{equation}\begin{aligned}
    \lim_{\alpha\to1} \wh{D}_\alpha(\rho\|\sigma) = \wh{D}(\rho\|\sigma) = \Tr \rho \log \left( \rho^{1/2} \sigma^{-1} \rho^{1/2} \right).
\end{aligned}\end{equation}
The Umegaki relative entropy has found many natural connections with tasks in quantum information processing, whereas the Belavkin--Staszewski relative entropy has thus far mostly found use as a bound that --- while useful in settings where the Umegaki relative entropy is difficult to apply --- is not in general tight.
Owing both to the initials of the authors who first proposed it, as well as to its seeming lack of operational meaning, $\wh{D}(\rho\|\sigma)$ is thus sometimes referred to as the BS~relative entropy.

%%%%%%%%%%%%%%%%%%%%%%%%%%%%%%%%%%%%%%%%%%%%%%%%%%%%%%%%%%%%%%%%%%%%%%%%%%%%%%%%%%%

\subsection{Maximal trace distance and maximal purified distance}
\label{sec:maximal_prelim}

The maximal extension of trace distance $\wh{T}$ was first studied by Matsumoto~\cite{matsumoto_2010-1,matsumoto_2018}. In those works he realised that, similarly to other maximal divergences, the distance can exhibit some non-standard behaviour: for example, for any pure state $\ket\phi$ it holds that $\wh{T}(\rho , \proj{\phi}) = 1 - 2^{-D_{\max}(\phi \| \rho)}$~\cite{matsumoto_2010-1}, implying in particular that $\wh{T}(\proj{\psi}, \proj{\phi}) = 1$ for all distinct pure states;  furthermore, in general $\wh{T}$ may not obey the triangle inequality~\cite[Proposition~18]{matsumoto_2010-1}. It is nevertheless more well behaved for full-rank quantum states, and in fact for a large portion of states it simply equals the standard trace distance~\cite[Sec.~8]{matsumoto_2018}. Understanding the precise differences between the two types of distance, and in particular their implications for operational problems, is the aim of our investigation.

While the maximal trace distance does not appear to admit an exact analytical expression in general, it has a convenient variational form given by~\cite[Sec.~8]{matsumoto_2018}
\begin{equation}\begin{aligned}\label{eq:matsu_tracedist}
    \wh{T}(\rho,\sigma) &= \frac12 \left(\Tr \rho + \Tr \sigma\right) - \sup \lset \Tr M \bar M \geq 0,\; \rho \geq M,\; \sigma \geq M \rset.
\end{aligned}\end{equation}
This can be compared with a variational expression for the standard trace distance in terms of the non-commutative minimal~\cite{helstrom_1969,holevo_1973,cheng_2023-1}:
\begin{equation}\begin{aligned}\label{eq:normal_tracedist}
    T(\rho,\sigma) &=\frac12 \left(\Tr \rho + \Tr \sigma\right)- \sup \lset \Tr M \bar \rho \geq M,\; \sigma \geq M \rset.
\end{aligned}\end{equation}
We see that the only difference between the expressions is the constraint $M \geq 0$ in~\eqref{eq:matsu_tracedist}.

It will be convenient for us to work with the \emph{generalised trace distance}, which is a valid metric also when $\sigma$ is a potentially subnormalised state with $\Tr \sigma \leq 1$~\cite{tomamichel_2016}. It is given by
\begin{equation}\begin{aligned}
    T_+(\rho,\sigma) &= \frac12 \left( \Tr(\rho-\sigma) + \norm{\rho-\sigma}{1} \right)\\
    &= \inf \lset \Tr P \bar P \geq 0,\; \rho - \sigma \leq P \rset,
\end{aligned}\end{equation}
with the second line being a rewriting of the variational expression in~\eqref{eq:normal_tracedist}. 
This is again the minimal extension of the corresponding classical quantity, and for $\Tr \sigma = \Tr \rho = 1$ it reduces to $T(\rho,\sigma)$. Using Matsumoto's variational form, we can obtain an analogous expression for the maximal generalised trace distance:
\begin{equation}\begin{aligned}\label{eq:matsu_generalisedtrace}
    \wh{T}_+(\rho,\sigma) &= \frac12 \Tr(\rho-\sigma) + \wh{T}(\rho,\sigma)\\
    &= \inf \lset \Tr P \bar P \geq 0,\; \rho - \sigma \leq P \leq \rho \rset.
\end{aligned}\end{equation}

Matsumoto also already studied the minimal extension of the fidelity, which corresponds to the maximal extension of the purified distance~\cite{matsumoto_2010-1,matsumoto_2018}. This can in fact be understood as a special case of the maximal R\'enyi divergences: since
\begin{equation}\begin{aligned}
    \wc{D}_{1/2}(\rho\|\sigma) = - \log F(\rho,\sigma),
\end{aligned}\end{equation}
we can write
\begin{equation}\begin{aligned}
    \wh{P}(\rho,\sigma) = \sqrt{ 1 - 2^{-\wh{D}_{1/2}(\rho\|\sigma)} }.
\end{aligned}\end{equation}
As long as $\Tr \rho = 1$, the generalised (subnormalised) variant of this metric has the same expression.

The equivalent of the Fuchs--van de Graaf equivalence relation applies also in this setting, in the sense that
\begin{equation}\begin{aligned}\label{eq:fvdg_maximal}
   \wh{T}(\rho, \sigma) \leq \wh{P}(\rho,\sigma) \leq \sqrt{2 \wh{T}(\rho,\sigma)}.
\end{aligned}\end{equation}
This can be seen simply by lifting the standard Fuchs--van de Graaf inequalities: applying~\eqref{eq:fvdg} to distributions $p$ and $q$ such that $\N(p) = \rho$ and $\N(q) = \sigma$, we can optimise each inequality over all channels $\N$ to obtain~\eqref{eq:fvdg_maximal}.

%%%%%%%%%%%%%%%%%%%%%%%%%%%%%%%%%%%%%%%%%%%%%%%%%%%%%%%%%%%%%%%%%%%%%%%%%%%%%%%%%%%
%%%%%%%%%%%%%%%%%%%%%%%%%%%%%%%%%%%%%%%%%%%%%%%%%%%%%%%%%%%%%%%%%%%%%%%%%%%%%%%%%%%
%%%%%%%%%%%%%%%%%%%%%%%%%%%%%%%%%%%%%%%%%%%%%%%%%%%%%%%%%%%%%%%%%%%%%%%%%%%%%%%%%%%

\section{Maximal smooth max-relative entropy}

\subsection{Definition and variational expressions}

A quantity that will form an important technical tool in our investigation is a new smooth variant of the max-relative entropy.
Conventional definitions of smooth max-relative entropy often define it in terms of an optimisation over states in an $\ve$-ball of trace distance; here we will distinguish the normalised and subnormalised variants, given respectively as
\begin{equation}\begin{aligned}
    \Dmax{T}{n} (\rho \| \sigma) &= \log \inf \lset \lambda \bar \rho' \leq \lambda \sigma,\; T(\rho,\rho') \leq \ve,\; \rho'\geq 0,\; \Tr \rho' = 1 \rset\\
    \Dmax{T}{s} (\rho \| \sigma) &= \log \inf \lset \lambda \bar \rho' \leq \lambda \sigma,\; T_+(\rho,\rho') \leq \ve,\; \rho'\geq 0,\; \Tr \rho' \leq 1 \rset,
\end{aligned}\end{equation}
Note that $\Dmax{T}{n}(\rho\|\sigma) = \max \big\{ 0, \Dmax{T}{s}(\rho\|\sigma) \big\}$~\cite[Lemma~1]{regula_2026}.
One can however consider another definition of a quantum smooth max-relative entropy: first define the classical smoothing as above, and then consider its minimal extension to quantum states~\cite{regula_2026}:
\begin{equation}\begin{aligned}
    \Dcmax{T}{s}(\rho\|\sigma) &\coloneqq \sup \lset  \Dmax{T}{s}(p\|q) \bar p = \M(\rho), \; q = \M(\sigma),\; \M \in \mathrm{CPTP} \rset\\
    &\hphantom{:}= \log \inf \lset \lambda \bar \Tr (\rho - \lambda \sigma)_+ \leq \ve \rset\\
&\hphantom{:}= \log \inf \lset \lambda \bar \rho \leq \lambda \sigma + Q,\; Q \geq 0,\; \Tr Q \leq \ve \rset,
\end{aligned}\end{equation}
and analogously for the normalised variant, denoted $\Dcmax{T}{n}[\ve,=]$. 
The equality in the second line here was shown in~\cite[Proposition~3]{regula_2026}, which established an equivalence between this smooth divergence and a quantity introduced earlier by~\cite{datta_2015} as an `information spectrum' variant of max-relative entropy.

We would now like to consider the maximal quantum extensions of $\Dmax{T}{s}$ and $\Dmax{T}{n}$. Although it may not be immediately obvious that such quantities should have a manageable expression, Matsumoto's proof technique used to study trace distance in~\cite{matsumoto_2018} can be adapted to establish variational expressions that will turn out to be extremely useful.

\begin{boxed}
\begin{proposition}[(Variational forms)]\label{prop:dmax_max_ext}
The maximal extension of the smooth max-relative entropy equals
\begin{equation}\begin{aligned}\label{eq:max_ext_eq1}
    \wh{D}^\ve_{\max} (\rho \| \sigma) &\coloneqq \inf \lset \Dmax{T}{s} (p \| q) \bar \N(p) = \rho,\; \N(q) = \sigma,\; \N \in \mathrm{CPTP} \rset\\
    &\hphantom{:}= \log \inf \lset \lambda \bar 0 \leq Z \leq \rho \leq \lambda \sigma + Z,\; \Tr Z \leq \ve \rset\\
    &\hphantom{:}= \log \sup \lset \frac{\Tr (A-B) \rho  - \ve}{\Tr A \sigma} \bar A,B \geq 0,\; A-B \leq \id \rset.
\end{aligned}\end{equation}
Alternatively, it can be understood as an optimisation over all states in the $\ve$-ball defined by the maximal extension of the generalised trace distance:
\begin{equation}\begin{aligned}\label{eq:max_ext_eq2}
    \wh{D}^\ve_{\max} (\rho \| \sigma)   %\inf \lset \lambda \bar \rho' \leq \lambda \sigma,\; \norm{\rho-\rho'}{M+} \leq \ve,\; \rho'\geq 0,\; \Tr \rho' \leq 1 \rset,
    &= \inf \lset D_{\max}(\rho' \| \sigma) \bar \wh{T}_+(\rho,\rho') \leq \ve,\; \rho'\geq 0,\; \Tr \rho' \leq 1 \rset.
\end{aligned}\end{equation}
The corresponding variant that only allows a smoothing over normalised states is
\begin{equation}\begin{aligned}\label{eq:max_ext_eq3}
   \wh{D}^{\ve,=}_{\max}(\rho\|\sigma)  &= \inf \lset \Dmax{T}{n} (p \| q) \bar \N(p) = \rho,\; \N(q) = \sigma,\; \N \in \mathrm{CPTP} \rset\\
    &= \max \Big\{ 0,\, \wh{D}^\ve_{\max} (\rho \| \sigma) \Big\}\\
    &= \inf \lset D_{\max}(\rho' \| \sigma) \bar \wh{T}(\rho,\rho') \leq \ve,\; \rho'\geq 0,\; \Tr \rho' = 1 \rset.
\end{aligned}\end{equation}
\end{proposition}
\end{boxed}
\begin{remark}
One notices again that the only difference between the optimisation problems defining $\wc{D}^\ve_{\max}$ and $\wh{D}^\ve_{\max}$ is the additional constraint $Z \leq \rho$ in~\eqref{eq:max_ext_eq1}.

Examples of states such that $\wh{D}^\ve_{\max}  > \Dmax{T}{s} > \wc{D}^\ve_{\max}$ can be easily found numerically; however, there also exist non-commuting states for which the quantities are all equal, so general (in)equality conditions appear non-trivial to establish.

Just as $\wc{D}^\ve_{\max}$ can be understood as an inverse function of the hockey-stick divergence $E_\gamma(\rho\|\sigma) = \Tr (\rho - \gamma \sigma)_+$, $\wh{D}^\ve_{\max}$ can be understood as an inverse function of the maximal equivalent of the hockey stick divergence:%(could this be the real `high-stick' divergence?)
\begin{equation}\begin{aligned}
    \wh{E}_\gamma (\rho \| \sigma) = \inf \lset \Tr Z \bar  Z\geq 0, \; \rho - \gamma \sigma \leq Z \leq \rho \rset.
\end{aligned}\end{equation}
\end{remark}

\begin{proof}[Proof of Proposition~\ref{prop:dmax_max_ext}]
Consider any reverse test $(\N, p, q)$, that is, probability distributions  $p, q$ on a finite alphabet $\mathcal{X}$ and a channel $\N$ such that $\N(p) = \rho$ and $\N(q) = \sigma$. Take $\lambda$ to be any feasible solution for $\Dmax{T}{s}(p\|q) = \wc{D}^\ve_{\max}(p\|q)$, namely, a value of $\lambda$ such that $\tr ( p - \lambda q)_+ \leq \ve$.
We now adapt a construction introduced in~\cite[Section~8.1]{matsumoto_2018}. For some choice of $Z \geq 0$, define a modified reverse test $(\N', p', q')$ with $p', q' : \{ 0, 1, 2\} \to [0,1]$ as follows:
\begin{alignat}{3}
    \N'(\proj{0}) &= \frac{\rho - Z}{\Tr (\rho - Z)},\quad  &\N'(\proj{1}) &= \frac{Z}{\Tr Z}, \quad &\N'(\proj{2}) &= \frac{\lambda \sigma + Z - \rho}{\Tr (\lambda \sigma + Z - \rho)},\nonumber\\
    p'(0) &= \Tr (\rho - Z),  &p'(1) &= \Tr Z, & p'(2) &= 0, \label{eq:construction}\\
    q'(0) &= \frac1\lambda \Tr (\rho - Z),  &q'(1) &= 0,  &q'(2) &= \frac{1}{\lambda} \Tr (\lambda \sigma + Z - \rho).\nonumber
\end{alignat}
One can notice that the conditions $\N'(p') = \rho$ and $\N'(q') = \sigma$ are indeed satisfied, and that $\N'$ is a valid channel as long as $\rho \geq Z$ and $\lambda \sigma + Z \geq \rho$. Furthermore, one easily verifies that
\begin{equation}\begin{aligned}
    \tr ( p' - \lambda q')_+ &= \sum_{i=0}^2 \big( p'(i) - \lambda q'(i)\big)_+\\
    &= 0 + \Tr Z + 0.
\end{aligned}\end{equation}
Putting everything together: if it were possible to choose $Z\geq 0$ satisfying $\rho \geq Z$, $\lambda \sigma + Z \geq \rho$, and $\Tr Z \leq \ve$, then this would mean that $\lambda$ is a feasible solution also for $\wc{D}^{\ve}_{\max}(p'\|q')$. In this way, we could argue that any $\lambda$ feasible for $\wc{D}^\ve_{\max} (p\| q)$ is also feasible for $\wc{D}^{\ve}_{\max}(p'\|q')$, implying that it suffices to restrict ourselves to reverse tests of the form~\eqref{eq:construction} rather than more general tests $(\N, p ,q)$. And indeed, this is the case: choosing $Z$ as
\begin{equation}\begin{aligned}
    Z = \sum_{x\in\mathcal{X}} \big(p(x) - \lambda q(x)\big)_+ \; \N(\proj{x}),
\end{aligned}\end{equation}
we see that $Z \leq \rho$ as $\N(p) = \rho$, $Z \geq \rho - \lambda \sigma$ as $\N(q) = \sigma$, and $ \Tr Z = \tr (p - \lambda q)_+ \leq \ve$ by assumption.

Altogether, we have thus shown that
\begin{equation}\begin{aligned}\label{eq:repeat}
    \wh{D}^\ve_{\max} (\rho \| \sigma) &= \inf \lset \wc{D}^{\ve}_{\max}(p'\|q') \bar p', q' \text{ as in~\eqref{eq:construction} } \rset\\
    &= \log \inf \lset \lambda \bar 0 \leq Z \leq \rho \leq \lambda \sigma + Z,\; \Tr Z \leq \ve \rset,
\end{aligned}\end{equation}
which is precisely the claimed expression. The third line of~\eqref{eq:max_ext_eq1} follows by strong duality.

For the second part of the Proposition, we consider the optimisation problem over smoothed states defined in~\eqref{eq:max_ext_eq2}. Recall Matsumoto's characterisation of the trace distance~\eqref{eq:matsu_tracedist} as
\begin{equation}\begin{aligned}
    \wh{T}_+(\rho,\rho') &= \inf \lset \Tr P \bar P \geq 0,\; \rho - \rho' \leq P \leq \rho \rset.
\end{aligned}\end{equation}
Consider then any feasible solution $Z$ for~\eqref{eq:repeat} and define $\rho' \coloneqq \rho - Z$. By construction, $\rho' \geq 0$, $\Tr \rho' \leq 1$, and $\rho' \leq \lambda \sigma$. Furthermore, from the optimisation in~\eqref{eq:matsu_generalisedtrace} we see that $\wh{T}_+(\rho,\rho') \leq \Tr Z \leq \ve$, making $\rho'$ a feasible subnormalised state in the $\ve$-ball of $\rho$. Conversely, consider any $\rho'$ which satisfies $\wh{T}_+(\rho,\rho') \leq \ve$ and $\rho' \leq \lambda \sigma$, and let $P$ be an optimal solution in~\eqref{eq:matsu_generalisedtrace}. Then clearly $Z \coloneqq P$ is a feasible solution for~\eqref{eq:repeat} with $\rho \leq \lambda \sigma + Z$. The feasible solutions thus have the same feasible optimal values, implying that the two optimisation problems are equivalent.

Finally, for the normalised variant, the subnormalised result in~\eqref{eq:max_ext_eq1} that we just proved shows that
\begin{equation}\begin{aligned}
    \max \Big\{ 0,\, \wh{D}^\ve_{\max} (\rho \| \sigma) \Big\} &= \max \bigg\{ 0,\, \inf \lset \Dmax{T}{s} (p \| q) \bar \N(p) = \rho,\; \N(q) = \sigma,\; \N \in \mathrm{CPTP} \rset \bigg\}\\
    &= \inf \lset \vphantom{\bigg\{}\max \Big\{ 0,\, \Dmax{T}{s} (p \| q) \Big\} \bar \N(p) = \rho,\; \N(q) = \sigma,\; \N \in \mathrm{CPTP} \rset\\
    &= \inf \lset \Dmax{T}{n} (p \| q) \bar \N(p) = \rho,\; \N(q) = \sigma,\; \N \in \mathrm{CPTP} \rset,
\end{aligned}\end{equation}
where the last equality is by~\cite[Lemma~1]{regula_2026}. The equivalence in the last line of~\eqref{eq:max_ext_eq3} can be shown by following the proof of~\cite[Lemma~1]{regula_2026} verbatim; for brevity we do not consider it here explicitly.
\end{proof}

%%%%%%%%%%%%%%%%%%%%%%%%%%%%%%%%%%%%%%%%%%%%%%%%%%%%%%%%%%%%%%%%%%%%%%%%%%%%%%%%%%%

\subsection{Asymptotic equipartition property}

Our main results will rely on the following asymptotic characterisation of the maximal smooth max-relative entropy.

\begin{boxed}
\begin{proposition}[(Maximal asymptotic equipartition)]\label{prop:dmax_aep}
For all $\ve \in (0,1)$, all $\alpha \in (0,1)$, and all $\beta \in (1,\infty)$, it holds that
\begin{equation}\begin{aligned}\label{eq:renyi_bound}
     \wh{D}_\alpha (\rho \| \sigma) - \frac{1}{1-\alpha} \log \frac{1}{1-\ve} \leq \wh{D}^\ve_{\max} (\rho \| \sigma)  \leq \wh{D}_\beta(\rho \| \sigma)  + \frac{1}{\beta-1} \log \frac{1}{\ve} - \log\frac{1}{1-\ve}
\end{aligned}\end{equation}
where $\wh{D}_\alpha$ denotes the maximal R\'enyi divergences.

\smallskip
As a result,
\begin{equation}\begin{aligned}\label{eq:aep}
    \lim_{n\to\infty} \frac1n \wh{D}^\ve_{\max} (\rho^{\otimes n} \| \sigma^{\otimes n}) =  \lim_{n\to\infty} \frac1n \wh{D}^{\ve,=}_{\max} (\rho^{\otimes n} \| \sigma^{\otimes n}) = \wh{D} (\rho \| \sigma),
\end{aligned}\end{equation}
where $\wh{D}$ is the Belavkin--Staszewski relative entropy.
\end{proposition}
\end{boxed}
\begin{proof}
For classical probability distributions, we have that~\cite[Cor.~11--12]{regula_2026}
\begin{equation}\begin{aligned}
    D_\alpha (p \| q) - \frac{1}{1-\alpha} \log \frac{1}{1-\ve} \leq D^{\ve,T,\leq}_{\max} (p \| q) \leq D_\beta(p \| q) + \frac{1}{\beta-1} \log \frac{1}{\ve} - \log \frac{1}{1-\ve}.
\end{aligned}\end{equation}
Minimising all terms over reverse tests $(\N,p,q)$ such that $\N(p) = \rho$ and $\N(q) = \sigma$ %and using the fact that the geometric R\'enyi divergences are the maximal extensions of classical R\'enyi divergences~\cite{matsumoto_2018} 
gives~\eqref{eq:renyi_bound}.

Applying this to $\rho^{\otimes n}$ and $\sigma^{\otimes n}$, dividing by $n$, and then taking the limit $n\to\infty$ gives
\begin{equation}\begin{aligned}
    \wh{D}_\alpha( \rho \| \sigma) \leq \lim_{n\to\infty} \frac1n \wh{D}^\ve_{\max} (\rho^{\otimes n} \| \sigma^{\otimes n}) \leq \wh{D}_\beta( \rho \| \sigma)
\end{aligned}\end{equation}
in the range $\alpha \in (0,1)$ and $\beta \in (1,2]$, which is the range in which we know that the maximal R\'enyi divergences equal the geometric R\'enyi divergences~\cite{matsumoto_2018} and hence are additive. 
Taking the limits $\alpha \to 1^-$ and $\beta \to 1^+$ concludes the proof.
\end{proof}

The lower bound in terms of $\wh{D}_\alpha$ is in fact asymptotically tight.

\begin{proposition}[(Strong converse exponent of maximal smooth $D_{\max}$)]\label{prop:strong_converse_exponent}
For any fixed $R>0$, let
\begin{equation}\begin{aligned}
    \ve_n(R) \coloneqq \inf \lset \ve \in (0,1) \bar \wh{D}^\ve_{\max}(\rho^{\otimes n} \| \sigma^{\otimes n}) \leq n R \rset.
\end{aligned}\end{equation}
Then for any states $\rho$ and $\sigma$ with $\supp(\rho)\subseteq\supp(\sigma)$, 
\begin{equation}\begin{aligned}
    \lim_{n\to\infty} -\frac1n \log \big( 1 - \ve_n(R) \big) = \sup_{\alpha \in (0,1)} (\alpha - 1) \left( R - \wh{D}_\alpha( \rho \| \sigma) \right).
\end{aligned}\end{equation}
\end{proposition}
\begin{proof}
The lower bound of Proposition~\ref{prop:dmax_aep} immediately gives that if $\wh{D}^{\ve_n}_{\max}(\rho^{\otimes n} \| \sigma^{\otimes n}) \leq n R $, then
\begin{equation}\begin{aligned}
   -\frac1n \log (1- \ve_n) \geq \frac1n (\alpha - 1) \left( n R - n \wh{D}_\alpha( \rho \| \sigma) \right)
\end{aligned}\end{equation}
for all $\alpha \in (0,1)$, using the additivity of $\wh{D}_\alpha$. Conversely, by definition we have that
\begin{equation}\begin{aligned}\label{eq:upper_matsumoto_exponent}
    \wh{D}^{\ve_n}_{\max}(\rho^{\otimes n} \| \sigma^{\otimes n}) \leq \Dmax{T}{s}[\ve_n] ( \overline{p}_n \| \overline{q}_n ), %= \wt{D}^{1-\ve_n}_{\max} ( \overline{p}^{\otimes n} \| \overline{q}^{\otimes n} ),
\end{aligned}\end{equation}
where $\overline{p}_n$, $\overline{q}_n$ denote the Matsumoto distributions~\cite{matsumoto_2010-1} defined as
\begin{equation}\begin{aligned}
    \overline{q}_n(x) = \braket{x | \sigma^{\otimes n} | x},\qquad \overline{p}_n(x) = \lambda_x q(x),
\end{aligned}\end{equation}
with $\lambda_x$ and $\ket{x}$ being the eigenvalues and corresponding eigenvectors of the operator $\Delta_n \coloneqq (\sigma^{\otimes n})^{-1/2} \rho^{\otimes n} (\sigma^{\otimes n})^{-1/2} = \sum_x \lambda_x \proj{x}$. Noticing that $\Delta_n = \Delta_1^{\otimes n}$, the Matsumoto distributions for $\rho^{\otimes n}$ and $\sigma^{\otimes n}$ are simply $\overline{p}_n = \overline{p}^{\otimes n}_1 \eqqcolon \overline{p}^{\otimes n}$ and $\overline{q}_n = \overline{q}^{\otimes n}_1 \eqqcolon \overline{q}^{\otimes n}$. 
Hence
\begin{equation}\begin{aligned}
     \wh{D}^{\ve_n}_{\max}(\rho^{\otimes n} \| \sigma^{\otimes n}) \leq \Dmax{T}{s}[\ve_n] ( \overline{p}^{\otimes n} \| \overline{q}^{\otimes n} ),
\end{aligned}\end{equation}
which implies that the optimal strong converse exponent cannot be larger than the optimal strong converse exponent of the classical quantity $\Dmax{T}{s}[\ve_n] ( \overline{p}^{\otimes n} \| \overline{q}^{\otimes n} ) = \wc{D}^{\ve_n}_{\max} ( \overline{p}^{\otimes n} \| \overline{q}^{\otimes n} )$. The latter is known to be 
\begin{equation}\begin{aligned}
    \sup_{\alpha \in (0,1)} (\alpha - 1) \big( R - D_\alpha( \overline{p} \| \overline{q}) \big),
\end{aligned}\end{equation}
which follows from standard classical large deviations results --- an explicit statement can be found in Theorem 2 of~\cite{salzmann_2022}. 
Then
\begin{equation}\begin{aligned}
    \lim_{n\to\infty} -\frac1n \log \big( 1 - \ve_n(R) \big) \leq \sup_{\alpha \in (0,1)} (\alpha - 1) \left( R - D_\alpha( \overline{p} \| \overline{q}) \right) = \sup_{\alpha \in (0,1)} (\alpha - 1) \left( R - \wh{D}_\alpha( \rho \| \sigma) \right),
\end{aligned}\end{equation}
with the equality due to the fact that the Matsumoto distributions are optimal for all $\wh{D}_\alpha$ with $\alpha \in (0,1)$~\cite{matsumoto_2018}.
\end{proof}

%%%%%%%%%%%%%%%%%%%%%%%%%%%%%%%%%%%%%%%%%%%%%%%%%%%%%%%%%%%%%%%%%%%%%%%%%%%%%%%%%%%
%%%%%%%%%%%%%%%%%%%%%%%%%%%%%%%%%%%%%%%%%%%%%%%%%%%%%%%%%%%%%%%%%%%%%%%%%%%%%%%%%%%
%%%%%%%%%%%%%%%%%%%%%%%%%%%%%%%%%%%%%%%%%%%%%%%%%%%%%%%%%%%%%%%%%%%%%%%%%%%%%%%%%%%

\section{Resource manipulation under maximal distance}
\label{sec:resources}

\subsection{Asymptotic manipulation}

We will consider the general setting of quantum resource theories. We are given here some set of \emph{free operations} $\OO$ which are, in the given setting, considered easy or cheap to implement.
The question of asymptotic manipulation of quantum states under the constraints of the given resource theory then asks about the best conversion rate under the free operations, that is,
\begin{equation}\begin{aligned}
    R_\Delta(\rho \xrightarrow[\OO]{} \sigma) \coloneqq \sup \lsetr R > 0 \barr \limsup_{n\to\infty} \inf_{\N \in \OO} \Delta\!\left(\N(\rho^{\otimes n}), \sigma^{\otimes \floor{Rn}} \right) = 0 \rsetr,
\end{aligned}\end{equation}
where $\Delta$ stands for a choice of distance measure: one typically sets $\Delta \in \{ T, P \}$, and we will also consider $\Delta \in \{ \wh{T}, \wh{P}\}$. 
Here, one assumes that for each $n$ there is a corresponding set of free operations $\OO_n$ in the given space; we do not notationally distinguish this and simply use $\OO$. 

One important feature in the definition of asymptotic rates is that, due to the Fuchs--van de Graaf inequalities~\eqref{eq:fvdg}, the trace distance can only converge to zero if so does the purified distance; that is,
\begin{equation}\begin{aligned}
    R_T(\rho \xrightarrow[\OO]{} \sigma) = R_P(\rho \xrightarrow[\OO]{} \sigma) 
\end{aligned}\end{equation}
for any pair of states and any set of free operations. 
As discussed in Section~\ref{sec:maximal_prelim}, the same is true for the maximal distances, giving
\begin{equation}\begin{aligned}
    R_{\wh{T}}(\rho \xrightarrow[\OO]{} \sigma) = R_{\wh{P}}(\rho \xrightarrow[\OO]{} \sigma) .
\end{aligned}\end{equation}
This means that one does not need to distinguish the two distances in the asymptotic analysis, and we will often state rates only in terms of the trace distances $\Delta \in \{ T, \wh{T}\}$.

The quantitative study of quantum resources relies on many entropic quantities defined with respect to the set of \emph{free states} $\FF$ of the given theory, typically understood as the states that do not possess any resources and can be prepared using the free operations. 
In entanglement theory, $\FF = \rm SEP$ would be composed of all separable  states, while in the theory of thermodynamics it consists only of (copies of) the Gibbs state $\gamma$. 
For any divergence $\DD$, we define
\begin{equation}\begin{aligned}
    \DD(\rho \| \FF) \coloneqq \inf_{\sigma \in \FF} \DD(\rho \| \sigma).
\end{aligned}\end{equation}
Due to the possible non-additivity of the resulting optimised divergences, the asymptotic rates of resource manipulation are often connected with regularised entropies of the form
\begin{equation}\begin{aligned}
    \DD^\infty(\rho \| \FF) \coloneqq \liminf_{n\to\infty} \frac1n \DD(\rho^{\otimes n} \| \FF).
\end{aligned}\end{equation}
This again implicitly assumes that a corresponding set of free states $\FF_n$ exists in each $n$-copy space; we use $\FF$ for simplicity without distinguishing this explicitly.

The study of quantum resource manipulation is often divided into resource distillation and resource dilution --- these two structurally very different problems exhibit strong differences in the techniques needed for their precise analysis.

%%%%%%%%%%%%%%%%%%%%%%%%%%%%%%%%%%%%%%%%%%%%%%%%%%%%%%%%%%%%%%%%%%%%%%%%%%%%%%%%%%%

\subsection{Distillation}

The task of distillation is concerned with converting a given noisy quantum state $\rho$ into copies of a pure state $\phi_\star$, typically representing the maximally resourceful state of the given theory. For example, in entanglement distillation the target state $\phi_\star$ would be given by the two-qubit maximally entangled state $\Phi_{AB}^+$, while in quantum thermodynamics it would represent a `work bit' $\phi_\gamma$, namely the minimal eigenstate of the Gibbs state $\gamma$. 

It will be convenient to define the one-shot distillable resource up to error $\ve$ with respect to a choice of distance $\Delta \in \{ T, P, \wh{T} \}$ as
\begin{equation}\begin{aligned}
    \mathrm{Distill}_{\OO}^{\ve,\Delta}(\rho) \coloneqq \sup \lsetr k \barr \inf_{\N \in \OO} \Delta(\N(\rho), \phi_\star^{\otimes k}) \leq \ve \rsetr.
\end{aligned}\end{equation}
With this, the asymptotic rate of distillation is given by
\begin{equation}\begin{aligned}
 R_\Delta(\rho \xrightarrow[\OO]{} \phi_\star) = \liminf_{\ve\to 0} \liminf_{n \to \infty} \frac1n \mathrm{Distill}_{\OO}^{\ve,\Delta}(\rho^{\otimes n}).
\end{aligned}\end{equation}

We will find that, in all resource theories considered here, the rate of distillation does not change regardless of whether standard or maximal distances are used to measure the error. 
The basic idea for why it is so is that the target pure state $\phi_\star$ in these theories obeys a number of symmetries; one can then employ \emph{twirling} at the end of each distillation protocol, reducing the optimisation over such protocols to a type of isotropic-like twirled states. And for such states, as we now show, the maximal distance simply reduces to the conventional trace distance.

\begin{lemma}\label{lem:isotropic}
Consider an `isotropic' state of the form
\begin{equation}\begin{aligned}
    \omega_p = (1-p) \, \phi + p \tau
\end{aligned}\end{equation}
where $p \in [0,1]$, $\phi$ is a pure state, and $\tau$ is a state orthogonal to $\phi$. Then
\begin{equation}\begin{aligned}
    \wh{T}(\omega_p ,\phi) = T(\omega_p , \phi) = P(\omega_p, \phi)^2 = p.
\end{aligned}\end{equation}
\end{lemma}
\begin{proof}
On the one hand, the orthogonality of $\phi$ and $\tau$ allows us to compute 
\begin{equation}\begin{aligned}
   T(\omega_p, \phi) = %\frac12\norm{\omega_p - \phi}{1} = 
   \frac12 \norm{- p \phi + p \tau }{1} = \frac12 (p + p) = p
\end{aligned}\end{equation}
and
\begin{equation}\begin{aligned}
    P(\omega_p, \phi) = \sqrt{1 - \Tr \omega_p \phi} = \sqrt{p}.
\end{aligned}\end{equation}
On the other hand,
\begin{equation}\begin{aligned}
    \omega_p - \phi \leq p \tau \leq p \tau + (1-p) \phi = \omega_p,
\end{aligned}\end{equation}
so the variational form of $\wh{T}$ in~\eqref{eq:matsu_generalisedtrace} gives $\wh{T}(\omega_p ,\phi) \leq \Tr (p \tau) = p$. Since $\wh{T}(\omega_p ,\phi) \geq T(\omega_p ,\phi)$ by definition, the distances must thus be equal.
\end{proof}

This immediately gives us the following general result.

\begin{boxed}
\begin{proposition}[(Distillation under maximal distances)]\label{prop:distillation_twirl}
Let $\phi_\star$ be the target pure state of distillation in the given resource theory. Assume that, for every $k \in \NN$, there exists a twirling channel $\T_k$ such that
\begin{enumerate}[(i)] 
\item its output is isotropic: for any input state $\rho$, there exists a $p \in [0,1]$ such that
\begin{equation}\begin{aligned}
    \T_k(\rho) = (1-p)\, \phi_\star^{\otimes k} + p \tau_k, \qquad \Tr\!\left[ \phi_\star^{\otimes k} \tau_k \right] = 0;
\end{aligned}\end{equation}
\item copies of the target state $\phi_\star$ are invariant under its action, i.e.\ $\T_k(\phi_\star^{\otimes k}) = \phi_\star^{\otimes k}$;
\item for any free operation $\N \in \OO$, $\T_k \circ \N$ is also a free operation.
\end{enumerate}
Then the one-shot distillable resource under standard trace distance, under maximal trace distance, and under squared purified distance are all the same:
\begin{equation}\begin{aligned}
    \mathrm{Distill}_{\OO}^{\ve,T}(\rho)  = \mathrm{Distill}_{\OO}^{\ve,\wh{T}}(\rho) = \mathrm{Distill}_{\OO}^{\sqrt\ve,P}(\rho).
\end{aligned}\end{equation}
As a result, the rate of distillation is the same for any distance among $\{ T, P, \wh{T}, \wh{P} \}$:
\begin{equation}\begin{aligned}
     R_{\wh{T}}(\rho \xrightarrow[\OO]{} \phi_\star) =  R_T(\rho \xrightarrow[\OO]{} \phi_\star) = R_{\wh{P}}(\rho \xrightarrow[\OO]{} \phi_\star) = R_{P}(\rho \xrightarrow[\OO]{} \phi_\star).
\end{aligned}\end{equation}
\end{proposition}
\end{boxed}

\begin{remark}
Many resource theories of interest immediately satisfy the assumptions of the theorem: examples are quantum entanglement (with $\OO$ standing for any class of free operations that includes one-way LOCC, which is sufficient to implement isotropic twirling~\cite{horodecki_1999-1}), quantum thermodynamics (with $\OO$ denoting the general class of Gibbs-preserving operations or the more constrained Gibbs-preserving covariant operations), or quantum coherence (with $\OO$ being any class of operations that includes incoherent operations).
\end{remark}

\begin{proof}
Since $T \leq \wh{T}$, any protocol that achieves error $\ve$ with respect to the maximal trace distance is also a feasible protocol for the standard trace distance, so $\mathrm{Distill}_{\OO}^{\ve,\wh{T}}(\rho) \leq \mathrm{Distill}_{\OO}^{\ve,T}(\rho)$ for any $\ve$.  On the other hand, let $\N \in \OO$ be any feasible free operation such that $T(\N(\rho), \phi_\star^{\otimes k}) \leq \ve$. By assumption, $\T_k \circ \N$ is a free operation, and by Lemma~\ref{lem:isotropic} we know that
\begin{equation}\begin{aligned}
    \wh{T}(\T_k \circ \N(\rho), \phi_\star^{\otimes k}) &= T(\T_k \circ \N(\rho), \phi_\star^{\otimes k})\\
    &=  T(\T_k \circ \N(\rho), \T_k(\phi_\star^{\otimes k}))\\
    &\leq T(\N(\rho), \phi_\star^{\otimes k})\\
    &\leq \ve,
\end{aligned}\end{equation}
with the third line by data processing. The channel $\T_k \circ \N$ is thus a feasible distillation protocol, giving $\mathrm{Distill}_{\OO}^{\ve,T}(\rho) \leq \mathrm{Distill}_{\OO}^{\ve,\wh{T}}(\rho)$. An analogous argument via Lemma~\ref{lem:isotropic} shows an equivalence with the purified distance too. 
\end{proof}

Since it is known that the rates of distillation in many quantum resources are governed by the Umegaki quantum relative entropy, this property then extends to rates defined in maximal distance. The example cases of interest are for example quantum entanglement distillation under non-entangling operations (NE)~\cite{brandao_2011,hayashi_2025-3,lami_2025-1}
\begin{equation}\begin{aligned}
    %\mathrm{Distill}_{\rm NE}^{\wh{T},\infty}(\rho_{AB}) 
    R_{\wh{T}}(\rho \xrightarrow[\rm NE]{} \Phi^+_{AB}) = D^\infty(\rho_{AB}\|{\rm SEP}),
\end{aligned}\end{equation}
and the simpler formula for quantum entanglement distillation from pure states under (one-way) local operations and classical communication~\cite{bennett_1996-2}
\begin{equation}\begin{aligned}
     %\mathrm{Distill}_{\rm LOCC}^{\wh{T},\infty}(\psi_{AB}) 
     R_{\wh{T}}(\psi_{AB} \xrightarrow[\rm LOCC]{} \Phi^+_{AB}) &= D(\psi_{AB}\|{\rm SEP})\\
     &=  S(\psi_A).
 \end{aligned}\end{equation} 
Another relevant example is quantum thermodynamic work extraction under Gibbs-preserving operations (GPO)~\cite{brandao_2013,wang_2019,buscemi_2019}
\begin{equation}\begin{aligned}
    %\mathrm{Distill}_{\rm GPO}^{\wh{T},\infty}(\rho) 
    R_{\wh{T}}(\rho \xrightarrow[\rm GPO]{} \phi_\gamma) = \frac{1}{\lambda} \,D(\rho\|\gamma),
\end{aligned}\end{equation}
where $\lambda = D(\phi_\gamma\|\gamma)$ is a normalisation factor. We note that definitions of extractable work in thermodynamic literature would be expressed in units of energy and would thus correspond not exactly to the rate $R_{\Delta}(\rho \xrightarrow[\OO]{} \phi_\gamma)$, but rather the rate multiplied by the factor $\frac{\lambda \ln2}{\beta}$ with $\beta$ standing for inverse temperature; since the two are equivalent up to a different normalisation, we do not consider it explicitly.

%%%%%%%%%%%%%%%%%%%%%%%%%%%%%%%%%%%%%%%%%%%%%%%%%%%%%%%%%%%%%%%%%%%%%%%%%%%%%%%%%%%

\subsection{Dilution}
\label{sec:resources_cost}

Resource dilution is a task opposite to distillation: we now start with the pure state $\phi_\star$, and aim to use it to prepare a general quantum state $\rho$. Treating copies of $\phi_\star$ as units of resource, the rate of this task is frequently termed `resource cost': it measures how expensive a given state is to prepare within the constraints of the given resource theory.  Precisely, we define the one-shot resource cost as
\begin{equation}\begin{aligned}
    \mathrm{Cost}_{\OO}^{\ve,\Delta}(\rho) \coloneqq \inf \lsetr k \barr \inf_{\N \in \OO} \Delta(\N(\phi_\star^{\otimes k}), \rho) \leq \ve \rsetr.
\end{aligned}\end{equation}
The asymptotic resource cost is then the reciprocal of the rate of converting $\phi_\star$ into $\rho$, that is,
\begin{equation}\begin{aligned}
\limsup_{\ve\to 0} \limsup_{n \to \infty} \frac1n \mathrm{Cost}_{\OO}^{\ve,\Delta}(\rho^{\otimes n}) =  R_\Delta(\phi_\star \xrightarrow[\OO]{} \rho)^{-1}.
\end{aligned}\end{equation}

The key insight for us is that, at the one-shot level, dilution in many quantum resource theories governed by variants of smooth max-relative entropy~\cite{brandao_2010,horodecki_2013,liu_2019}. Recalling the asymptotic equipartition property that we established in Proposition~\ref{prop:dmax_aep}, this already suggests that a change of the distance from trace or purified distance to their maximal counterparts could result in the Belavkin--Staszewski relative entropy emerging as the asymptotic limit of the task.

Formalising this precisely, however, requires some assumptions. We will first consider explicitly the examples of quantum resources that are of most interest to us here. In Appendix~\ref{app:general} we give a more general formulation that can be applied to broader quantum resources. Our approach here follows standard achievability results in one-shot quantum resource theories~\cite{brandao_2010,liu_2019}, adapted to account for the maximal distance criterion.

\begin{boxed}
\begin{proposition}[(Cost in thermodynamics)]
\label{prop:cost_thermo}
Consider the resource theory of thermodynamics, or more generally any resource theory where the only free states are of the form $\gamma^{\otimes k}$ for some state $\gamma$. Choose $\phi_\gamma$ to be the projector onto the eigenvector corresponding to the smallest (non-zero) eigenvalue of $\gamma$, and let $\lambda = D(\phi_\gamma \| \gamma)$. Then the one-shot cost under Gibbs-preserving operations measured with the maximal trace distance is
\begin{equation}\begin{aligned}
    \mathrm{Cost}_{\rm GPO}^{\ve,\wh{T}}(\rho) = \ceil{ \frac{1}{\lambda} \, \wh{D}^{\ve,=}_{\max}(\rho \| \gamma) }
\end{aligned}\end{equation}
for all states $\rho$. Consequently, the asymptotic cost is
\begin{equation}\begin{aligned}
    R_{\wh{T}}(\phi_\gamma \xrightarrow[\rm GPO]{} \rho)^{-1} = \frac{1}{\lambda}\, \wh{D}(\rho \| \gamma).
\end{aligned}\end{equation}
\end{proposition}
\end{boxed}
\begin{proof}
On the one hand, for any feasible dilution protocol $\N$ we have that
\begin{equation}\begin{aligned}
    k \lambda &= k \, D_{\max}(\phi_\gamma \| \gamma)\\
    &= D_{\max}(\phi_\gamma^{\otimes k} \| \gamma^{\otimes k})\\
    &\textgeq{(i)} D_{\max}(\N(\phi_\gamma^{\otimes k}) \| \N(\gamma^{\otimes k}))\\
    &\textgeq{(ii)} \inf \lset D_{\max}(\rho' \|\gamma) \bar \wh{T}(\rho, \rho') \leq \ve \rset\\
    &\texteq{(iii)} \wh{D}^{\ve,=}_{\max}(\rho \| \gamma)
\end{aligned}\end{equation}
where (i) uses the data processing inequality, and in (ii) we made use of the fact that $\N$ is a feasible dilution protocol, meaning that it maps $\phi_\gamma^{\otimes k}$ to a state $\rho'$ which is $\ve$-close to $\rho$, while at the same time mapping $\gamma^{\otimes k}$ to $\gamma$, with the latter standing for the Gibbs state in the same output space as $\rho$. The last equality in (iii) follows by the variational form in Proposition~\ref{prop:dmax_max_ext}. 

On the other hand, let $\mu = \wh{D}^{\ve,=}_{\max}(\rho \| \gamma) = D_{\max}(\rho' \| \gamma)$, where $\rho'$ is a state with $\wh{T}(\rho,\rho') \leq \ve$. 
For any $k \geq \mu / \lambda$, we have $\rho' \leq 2^\mu \gamma \leq 2^{k \lambda} \gamma$, and so the channel
\begin{equation}\begin{aligned}
    \N(X) \coloneqq \Tr[\phi_\gamma^{\otimes k} X] \, \rho' + \Tr[(\id-\phi_\gamma^{\otimes k}) X] \frac{2^{k \lambda} \gamma - \rho'}{2^{k \lambda} - 1}
\end{aligned}\end{equation}
satisfies $\N(\phi_\gamma^{\otimes k}) = \rho'$ as well as
\begin{equation}\begin{aligned}
    \N(\gamma^{\otimes k}) &= \left(\frac{1}{2^{\lambda}}\right)^k \, \rho' + \left(1 - \left(\frac{1}{2^{\lambda}}\right)^k\right) \frac{2^{k \lambda} \gamma - \rho'}{2^{k \lambda} - 1}\\
    &= \gamma,
\end{aligned}\end{equation}
making it a valid dilution protocol. This establishes the one-shot formula, from which the asymptotic expression follows by the equipartition property in Proposition~\ref{prop:dmax_aep}.
\end{proof}

Proposition~\ref{prop:cost_thermo} together with the distillation result in Proposition~\ref{prop:distillation_twirl} tell us that the reversibility of quantum thermodynamics is broken under the maximal trace distance.
\begin{boxed}
\begin{corollary}
Let $\OO \subseteq \rm GPO$ be any class of thermodynamic free operations, e.g.\ Gibbs-preserving or thermal operations. 
For any state $\rho$ which does not commute with $\gamma$,  we have
\begin{equation}\begin{aligned}
    R_{\wh{T}}(\phi_\gamma \xrightarrow[\OO]{} \rho)^{-1} &\geq \frac{1}{\lambda}\, \wh{D}(\rho \| \gamma)\\
    &> \frac{1}{\lambda}\, D(\rho \| \gamma)\\
    &\geq R_{\wh{T}}(\rho \xrightarrow[\OO]{} \phi_\gamma),
\end{aligned}\end{equation}
and so the resource theory of quantum thermodynamics is irreversible under maximal distance measures.
\end{corollary}
\end{boxed}
This follows as $D(\rho\|\gamma) < \wh{D}(\rho\|\gamma)$ for any non-commuting states~\cite[Example~4.4]{hiai_2017}. We also used that irreversibility under Gibbs-preserving operations forbids reversibility under any smaller class of operations: the resource cost can only increase, while the distillable resource can only decrease when the operations are restricted. 
The result directly applies also to the resource theory of asymmetric distinguishability~\cite{wang_2019}.

Another example of a theory in which reversibility no longer holds is quantum entanglement.
The reversibility of entanglement of pure states was famously shown in the early works of~\cite{bennett_1996-1,bennett_1996-2}. 
For general quantum states, the theory of entanglement is not reversible even under the maximal free set of non-entangling operations~\cite{lami_2023}; however, as first conjectured by~\cite{brandao_2010} and later proved by~\cite{hayashi_2025-3,lami_2025-1}, relaxing the assumptions on the allowed operations in a controlled way can indeed yield a reversible theory of entanglement under asymptotically non-entangling operations (ANE). 
One could follow this approach and connect the maximal-distance entanglement dilution under ANE with the regularised Belavkin--Staszewski relative entropy $\wh{D}^\infty(\rho\|\rm SEP)$. 
For simplicity, here we will focus on the case of pure states: we will see that even in this structurally simpler case, reversibility of entanglement is already broken. 

\begin{boxed}
\begin{proposition}[(Pure-state entanglement cost)]
\label{prop:cost_pureent}
For any pure state $\psi_{AB}$, the one-shot cost under non-entangling operations is
\begin{equation}\begin{aligned}
    \mathrm{Cost}_{\rm NE}^{\ve,\wh{T}}(\psi_{AB}) = \ceil{ \wh{D}^{\ve,=}_{\max}(\psi_{AB} \| \SEP) },
\end{aligned}\end{equation}
and its asymptotic limit is given by
\begin{equation}\begin{aligned}\label{eq:rate_NE}
    R_{\wh{T}}(\Phi_{AB}^+ \xrightarrow[\rm NE]{} \psi_{AB})^{-1} &= \wh{D}(\psi_{AB} \| \SEP) = S_{1/2}(\psi_A).
\end{aligned}\end{equation}
The asymptotic rate in~\eqref{eq:rate_NE} holds also under asymptotically non-entangling operations.
\end{proposition}
\end{boxed}
\begin{proof}
The converse direction of the proof is as in Proposition~\ref{prop:cost_thermo}. For the other direction, the simplification in the case of pure states stems from the fact that the generalised robustness of entanglement of a pure state equals the standard robustness~\cite{vidal_1999,steiner_2003,harrow_2003} --- specifically, the max-relative entropy of entanglement admits the alternative variational expression
\begin{equation}\begin{aligned}
    D_{\max}(\psi_{AB} \| {\rm SEP}) = \log \inf \lset \mu \bar \psi_{AB} + (\mu - 1) \sigma_- = \mu \sigma_+,\; \sigma_\pm \in {\rm SEP} \rset.
\end{aligned}\end{equation}
Choose then $\psi'_{AB} \coloneqq (1-\ve) \psi_{AB} + \ve \sigma_-$ and construct the map
\begin{equation}\begin{aligned}
 \N(X) &\coloneqq \Tr[\Phi_{A^kB^k}^+ X] \, \psi'_{AB} + \Tr[(\id-\Phi_{A^kB^k}^+) X] \sigma_-\\
 &= \Tr[\Phi_{A^kB^k}^+ X] (1-\ve) \psi_{AB} + \left( \Tr X - \Tr[\Phi_{A^kB^k}^+ X] (1-\ve)\right) \sigma_- .
    % \N(\Phi_{A^kB^k}^+) &= 
\end{aligned}\end{equation}
Since $\Tr[ \Phi_{A^kB^k}^+ \tau ] \leq \frac{1}{2^k}$ for any $\tau \in \rm SEP$,  $\N$ is non-entangling for any $k \geq  D_{\max}(\psi_{AB} \| {\rm SEP}) + \log (1-\ve)$. 
Noticing now that
\begin{equation}\begin{aligned}
    \wh{D}^\ve_{\max}(\psi_{AB} \| \SEP) &= \log \inf \lset \lambda \bar 0 \leq Z \leq \psi_{AB} \leq \lambda \sigma + Z,\; \sigma \in {\rm SEP},\; \Tr Z \leq \ve \rset\\
    &= \log \inf \lset \lambda \bar \psi_{AB} \leq \lambda \sigma + Z,\; \sigma \in {\rm SEP},\; Z = \ve \psi_{AB} \rset\\
    &= D_{\max}(\psi_{AB} \| \SEP) + \log (1-\ve)
\end{aligned}\end{equation}
concludes the proof of the one-shot expression for $\mathrm{Cost}_{\rm NE}^{\ve,\wh{T}}(\psi_{AB})$.

Since the target state here is pure, a major simplification occurs in the landscape of maximal divergences: for all $\alpha > 0$, the maximal R\'enyi divergences collapse to the max-relative entropy, i.e.\ $\wh{D}_\alpha(\psi \| \sigma) = D_{\max}(\psi \| \sigma)$~\cite{matsumoto_2018} (cf.~\cite[Proposition~71]{katariya_2021}). This means that in particular
\begin{equation}\begin{aligned}
    \wh{D}(\psi_{AB}^{\otimes n} \| {\rm SEP}) &= \wh{D}_\alpha (\psi_{AB}^{\otimes n} \| {\rm SEP})\\
    &= D_{\max}(\psi_{AB}^{\otimes n} \| {\rm SEP})\\
    &= n \, D_{\max} (\psi_{AB} \|{\rm SEP})\\
    &= n S_{1/2}(\psi_A),
\end{aligned}\end{equation}
using the known expression for the max-relative entropy of entanglement~\cite{vidal_1999,steiner_2003,harrow_2003}. Proposition~\ref{prop:dmax_aep} then implies the asymptotic formula.

To see that the asymptotic expression extends to the larger set of asymptotically non-entangling operations (ANE), recall from~\cite{brandao_2010} that such operations are defined as a sequence of maps $(\N_n)_n$ such that $D_{\max}(\N_n(\sigma)) \leq \delta_n$ for all $\sigma \in \rm SEP$, where $\lim_{n\to\infty} \delta_n \to 0$.
For any feasible sequence of ANE maps $(\N_n)_n$ such that $\wh{T}(\N_n((\Phi_{AB}^+)^{\otimes K_n}), \psi_{AB}^{\otimes n}) \leq \ve_n$, we then get a converse bound as
\begin{equation}\begin{aligned}
    K_n &= D_{\max}((\Phi_{AB}^+)^{\otimes K_n} \| \SEP)\\
    &\textgeq{(i)} D_{\max}(\N_n((\Phi_{AB}^+)^{\otimes K_n})\| \SEP) - \delta_n\\
    &\textgeq{(ii)} D_{\max}(\psi_{AB}^{\otimes n} \| \SEP) - \delta_n + \log(1-\ve_n)\\
    &= n S_{1/2}(\psi_A) - \delta_n + \log(1-\ve_n),
\end{aligned}\end{equation}
where (i) follows by~\cite[Lemma~IV.1]{brandao_2010}, which tells us that $D_{\max}(\cdot\|\SEP)$ is approximately monotonic under ANE operations, and (ii) follows since if a state $\omega$ satisfies that $\wh{T}(\psi,\omega) \leq \ve$ for a pure $\psi$ then the variational form of $\wh{T}$ in~\eqref{eq:matsu_generalisedtrace} gives $\omega \geq (1-\ve) \psi$. Dividing by $n$ and taking the limit tells us that the rate of $S_{1/2}(\psi_A)$ cannot be exceeded even under ANE.
\end{proof}

The announced irreversibility of entanglement is then a direct consequence.

\begin{boxed}
\begin{corollary}
Let $\OO \subseteq \rm ANE$ be any class of free operations, e.g.\ non-entangling maps or LOCC.
For any pure state $\psi_{AB}$ that is not maximally entangled, it holds that
\begin{equation}\begin{aligned}
 R_{\wh{T}}(\Phi_{AB}^+ \xrightarrow[\OO]{} \psi_{AB})^{-1} &\geq S_{1/2}(\psi_A)\\
 &> S(\psi_A)\\
 &\geq R_{\wh{T}}(\psi_{AB} \xrightarrow[\OO]{} \Phi_{AB}^+), 
 \end{aligned}\end{equation}
 and so the resource theory of entanglement is irreversible under maximal distance measures.
\end{corollary}
\end{boxed}

\begin{remark}[(One-shot dilution vs.\ distillation)]
Whereas the smooth max-relative entropy finds applications in the description of many quantum resources, one-shot distillation is often connected with the hypothesis testing relative entropy $D^\ve_H(\rho\|\sigma) = - \log \inf \lset \Tr M \sigma \bar 0 \leq M \leq \id,\; \Tr M \rho \geq 1- \ve \rset$ \cite{brandao_2010,liu_2019,wang_2019,regula_2020}, which quantifies the asymmetric error of quantum state discrimination. This quantity can also be expressed as the minimal quantum extension of the smooth R\'enyi divergence $D^\ve_0$~\cite{regula_2026-1}. One could then ask about the maximal extension of the hypothesis testing relative entropy: following Matsumoto's approach as we did in Proposition~\ref{prop:dmax_max_ext}, one obtains
\begin{equation}\begin{aligned}\label{eq:dh}
    \wh{D}^{\ve}_H (\rho \| \sigma) &\coloneqq \inf \lset D^\ve_H (p \| q) \bar \N(p) = \rho,\; \N(q) = \sigma,\; \N\in \mathrm{CPTP} \rset\\
    &\hphantom{:}= -\log \inf \lset \Tr A \sigma \bar  A,B\geq 0,\; A-B\leq \id,\; \Tr (A-B) \rho \geq 1-\ve \rset.
\end{aligned}\end{equation}
However, while the conventional $D^\ve_H$ is explicitly connected with the practical setting of quantum state discrimination, the maximal $\wh{D}^{\ve}_H$ requires a relaxation of the operational assumptions, and appears to depart from measurements that could be allowed in practice. We stress that distillation of quantum resources under maximal distances does \emph{not} use this quantity --- our Proposition~\ref{prop:distillation_twirl} shows that the conventional hypothesis testing relative entropy $D^\ve_H$ remains the relevant quantity. Indeed, one would not expect the maximal extension $\wh{D}^{\ve}_H$ to play a role here, as it is a \emph{relaxation} of the operational setting, while the approach of our work is based on imposing a \emph{more demanding} error criterion on the tasks in question, so the operational roles of $\wh{D}^\ve_H$ and $\wh{D}^\ve_{\max}$ go in opposite directions.
\end{remark}

%%%%%%%%%%%%%%%%%%%%%%%%%%%%%%%%%%%%%%%%%%%%%%%%%%%%%%%%%%%%%%%%%%%%%%%%%%%%%%%%%%%
%%%%%%%%%%%%%%%%%%%%%%%%%%%%%%%%%%%%%%%%%%%%%%%%%%%%%%%%%%%%%%%%%%%%%%%%%%%%%%%%%%%
%%%%%%%%%%%%%%%%%%%%%%%%%%%%%%%%%%%%%%%%%%%%%%%%%%%%%%%%%%%%%%%%%%%%%%%%%%%%%%%%%%%

\section{Quantum communication}\label{sec:comm}

We consider the theory of assisted quantum communication, which can be understood as the conversion of quantum channels under coding protocols that may use some additional assistance --- conventionally, shared randomness, shared entanglement, or shared no-signalling correlations are frequently studied in this context. We will in particular focus on the setting of shared correlations because it leads to the reverse quantum Shannon theorem~\cite{bennett_2014,berta_2011}, which states that channel transformations are asymptotically reversible: the amount of noiseless communication needed to produce a given channel (= simulation cost) exactly equals the amount of noiseless communication that can be extracted from a given channel (= channel capacity).

%%%%%%%%%%%%%%%%%%%%%%%%%%%%%%%%%%%%%%%%%%%%%%%%%%%%%%%%%%%%%%%%%%%%%%%%%%%%%%%%%%%

\subsection{Capacities}

Given a type of assistance denoted by $\Omega$, let $\C_\Omega$ denote the set of all $\Omega$-assisted coding schemes, understood as superchannels $(A \to B) \to (K \to K)$. The conventional $\Omega$-assisted one-shot quantum capacity of a channel $\N$ is defined as
\begin{equation}\begin{aligned}
    Q_{\Omega}^{\ve} (\N) \coloneqq \sup \lsetr \log |K| \barr \Diamond(\Theta[\N],  \idc_K) \leq \ve,\; \Theta \in \C_\Omega \rsetr,
\end{aligned}\end{equation}
where $\idc_K$ is the $K$-dimensional identity channel and
\begin{equation}\begin{aligned}
    \Diamond(\N,\M) &\coloneqq \frac12 \dia{\N - \M}\\
    &\hphantom{:}= \sup_{\psi_{RA}} \, T(\idc_R \otimes \N (\psi), \idc_R \otimes \M (\psi))
\end{aligned}\end{equation}
denotes the diamond norm distance. Here, the optimisation is over all pure states acting on the space $R \otimes A$ with $R \cong A$, which we will denote as $\psi_{AA}$.
The asymptotic capacity is then
\begin{equation}\begin{aligned}
    Q_{\Omega} (\N) \coloneqq \lim_{\ve\to0} \liminf_{n\to\infty} \frac1n Q_{\Omega}^{\ve} (\N^{\otimes n}).
\end{aligned}\end{equation}
Of particular importance to us will be the entanglement-assisted quantum capacity $Q_{\rm EA}$ and the no-signalling--assisted quantum capacity $Q_{\rm NS}$. Note that in both of these settings the classical capacity simply equals twice the quantum capacity~\cite{leung_2015}.

We will now consider quantum channel manipulation under the modified notion of distance, induced by the maximal trace distance. Specifically, for two channels $\N, \M : A \to B$, define
\begin{equation}\begin{aligned}
    \wh{\Diamond}(\N,\M) \coloneqq \sup_{\psi_{AA}} \, \wh{T}(\idc_A \otimes \N (\psi), \idc_A \otimes \M (\psi)),
\end{aligned}\end{equation}
where we recall that $\wh{T}(\cdot)$ is defined as in~\eqref{eq:matsu_tracedist}. 

Generalising the variational formula of Matsumoto for the maximal trace distance gives a variational expression for this quantity.
\begin{lemma}\label{lem:diamond_sdp}
For all channels $\N, \M : A \to B$ we have
\begin{equation}\begin{aligned}
   \wh{\Diamond}(\N,\M) = \inf \lset \norm{\Tr_B Z_{AB}}{\infty} \bar J_\N - J_\M \leq Z_{AB} \leq J_\N,\; Z_{AB} \geq 0 \rset
\end{aligned}\end{equation}
where $J_\N \coloneqq |A| \idc_A \otimes \N(\Phi_A^+)$ is the (unnormalised) Choi operator of the given channel.
\end{lemma}

Define now the $\Omega$-assisted one-shot quantum capacity of a channel $\N$ measured with the maximal distance $\wh{\Diamond}$ as
\begin{equation}\begin{aligned}
    \wh{Q}_{\Omega}^{\ve} (\N) \coloneqq \sup \lsetr \log |K| \barr \wh{\Diamond}(\Theta[\N],  \idc_K) \leq \ve,\; \Theta \in \C_\Omega \rsetr,
\end{aligned}\end{equation}
and analogously
\begin{equation}\begin{aligned}
    \wh{Q}_{\Omega} (\N) \coloneqq \lim_{\ve\to0} \liminf_{n\to\infty} \frac1n \wh{Q}_{\Omega}^{\ve} (\N^{\otimes n}).
\end{aligned}\end{equation}

Note that in the asymptotic definition $\wh{Q}_{\Omega} (\N)$, we could have also used the maximal extension of the purified distance instead of  $\wh{\Diamond}$ --- due to the Fuchs--van de Graaf inequalities, the rates induced by these distances are equivalent, as discussed in Section~\ref{sec:resources}.

\begin{boxed}
\begin{proposition}\label{prop:capacities}
Let $\Omega$ denote any type of assistance that includes shared randomness. Then, for all channels $\N$ and all $\ve \in [0,1]$, the one-shot capacity does not change under the maximal distance:
\begin{equation}\begin{aligned}
    Q_{\Omega}^{\ve} (\N) = \wh{Q}_{\Omega}^{\ve} (\N).
\end{aligned}\end{equation}
As a consequence, we immediately have that
\begin{equation}\begin{aligned}
    \wh{Q}_{\rm EA}(\N) = Q_{\rm EA} (\N) = \wh{Q}_{\rm NS} (\N) = Q_{\rm NS} (\N) = \frac12 I(\N)
\end{aligned}\end{equation}
where
\begin{equation}\begin{aligned}
    I(\N) = \inf_{\sigma_B} \sup_{\psi_{AA}} D(\idc_A \otimes \N (\psi_{AA}) \| \psi_A \otimes \sigma_B)
\end{aligned}\end{equation}
is the (Umegaki) mutual information of the channel. The same applies to strong converse capacities. 
\end{proposition}
\end{boxed}
\begin{proof}
A standard argument based on twirling~\cite{horodecki_1999} (which only needs shared randomness to implement) tells us that we can restrict ourselves to protocols which result in a linear combination of the identity channel $\idc_K$ and a channel orthogonal to it given by $\frac{|K|^2 \D_K - \idc_K }{|K|^2 - 1}$ where $\D_K : K \to K$ denotes the completely depolarising channel. Thus, the overall channel after the coding protocol can be assumed to be of the form
\begin{equation}\begin{aligned}
  J_{\Theta[\N]} = p  |K| \,\Phi_{K}^+ + (1-p) \frac{ \id - \Phi_K^+ }{|K|-\frac{1}{|K|}}
\end{aligned}\end{equation}
for some $p \in [0,1]$, where $\Phi_K^+$ is the maximally entangled state.

The crucial property of such channels, similarly to what we observed earlier in Lemma~\ref{lem:isotropic}, is that $\Diamond(\Theta[\N], \idc_K) = \wh\Diamond(\Theta[\N], \idc_K)$, meaning that the error does not change under the change of distance to the maximal one. 
To see this, notice on the one hand that
\begin{equation}\begin{aligned}
    \Diamond(\Theta[\N], \idc_K) &\geq \frac12\norm{ \idc_R \otimes \Theta[\N] (\Phi_{K}^+) - \Phi_K^+}{1}\\
    &= \frac12 \norm{ (p-1) \Phi_K^+ + (1-p) \frac{\id - \Phi_K^+}{|K|^2-1}}{1}\\
    &= 1-p
\end{aligned}\end{equation}
due to the orthogonality of $\Phi_K^+$ and $\frac{\id - \Phi_K^+}{|K|^2-1}$. On the other hand, notice that
\begin{equation}\begin{aligned}
    J_{\Theta[\N]} - J_{\idc_K}  \leq (1-p) \frac{ \id - \Phi_K^+ }{|K|-\frac{1}{|K|}} \leq J_{\Theta[\N]}.
\end{aligned}\end{equation}
Hence, using the variational form in Lemma~\ref{lem:diamond_sdp} we bound
\begin{equation}\begin{aligned}
    \wh{\Diamond}(\Theta[\N], \idc_K) \leq (1-p) \norm{\id}{\infty} = 1-p.
\end{aligned}\end{equation}
Since $\wh{\Diamond} \geq {\Diamond}$ by definition, the two distances must therefore be equal in this case.

The formulas for entanglement- and non-signalling--assisted capacities under standard distances follow from~\cite{bennett_2002,leung_2015}.
\end{proof}

%%%%%%%%%%%%%%%%%%%%%%%%%%%%%%%%%%%%%%%%%%%%%%%%%%%%%%%%%%%%%%%%%%%%%%%%%%%%%%%%%%%

\subsection{Channel simulation}

We define
\begin{equation}\begin{aligned}
    S_{\Omega}^{\ve} (\N) &\coloneqq \inf \lsetr \log |K| \barr \Diamond(\Theta[\idc_K] , \N) \leq \ve,\; \Theta \in \C_\Omega \rsetr,\\
    \wh{S}_{\Omega}^{\ve} (\N) &\coloneqq \inf \lsetr \log |K| \barr \wh{\Diamond}(\Theta[\idc_K], \N) \leq \ve,\; \Theta \in \C_\Omega \rsetr,
\end{aligned}\end{equation}
together with the corresponding asymptotic rates
\begin{equation}\begin{aligned}
    S_{\Omega} (\N) \coloneqq \lim_{\ve\to0} \limsup_{n\to\infty} \frac1n {S}_{\Omega}^{\ve} (\N^{\otimes n}), \qquad  \wh{S}_{\Omega} (\N) \coloneqq \lim_{\ve\to0} \limsup_{n\to\infty} \frac1n \wh{S}_{\Omega}^{\ve} (\N^{\otimes n}).
\end{aligned}\end{equation}
In the conventional setting, it is known that the simulation cost under no-signalling assistance is given by the smooth max-relative entropy distance to the set of replacer channels $\Rep $~\cite{fang_2020-1}
\begin{equation}\begin{aligned}
    S_{\rm NS}^{\ve} (\N) = \frac12 \Dmax{dia}{n} (\N \| \Rep) \coloneqq \frac12 \inf \lsetr D_{\max}(\N' \| \R_\sigma ) \barr \Diamond(\N, \N') \leq \ve,\; \R_\sigma \in \Rep \rsetr,
\end{aligned}\end{equation}
where $\Rep \coloneqq \lset \R : A \to B \bar \exists \textrm{ state } \sigma  \text{ s.t. } \R(\rho) = \sigma \; \forall \textrm{ states } \rho  \rset$ and the max-relative entropy between channels is defined as~\cite{wilde_2020}
\begin{equation}\begin{aligned}
D_{\max}(\N \| \M) &= \sup_{\psi_{AA}} D_{\max}( \idc\otimes \N(\psi_{AA}) \| \idc\otimes \M(\psi_{AA}))\\
&= D_{\max} (J_\N \| J_\M).
\end{aligned}\end{equation}
We show that an analogous argument holds for the maximal distance.
\begin{lemma}\label{lem:oneshot_cost_ns}
\begin{equation}\begin{aligned}
    \wh{S}_{\rm NS}^{\ve} (\N) = \frac12 \Dmax{mdia}{n} (\N \| \Rep) \coloneqq \frac12 \inf \lsetr D_{\max}(\N' \| \R_\sigma ) \barr \wh{\Diamond}(\N, \N') \leq \ve,\; \R_\sigma \in \Rep \rsetr.
\end{aligned}\end{equation}
\end{lemma}
\begin{proof}
The key point is that a superchannel $\Theta$ is no-signalling if and only if it preserves the set of replacer channels~\cite[Proposition~1]{takagi_2020}. One direction (converse) of the lemma is then a consequence of the data processing inequality for $D_{\max}$: the inequalities
\begin{equation}\begin{aligned}
    \log |K|^2 &= D_{\max}(\idc_K \| \Rep)\\
    &\geq D_{\max}(\Theta[\idc_K] \| \Rep)\\
    &\geq \Dmax{mdia}{n} (\N \| \Rep)
\end{aligned}\end{equation}
hold for any protocol such that $\wh{\Diamond}(\Theta[\idc_K], \N) \leq \ve$.

For the other direction, pick any $\N'$ such that $\wh{\Diamond}(\N, \N')\leq \ve$, and choose any $\lambda$ such that $\N' \leq \lambda \R_\sigma$ for some replacer channel $\R_\sigma$. Construct the superchannel $\Theta: (K \to K) \to (A \to B)$ as
\begin{equation}\begin{aligned}
    \Theta [ \K ] \coloneqq  \frac{\Tr  \Phi_+ J_\K }{|K|} \, \N' + \frac{ \Tr  (\id-\Phi_+) J_\K  }{|K|} \, \frac{ \lambda \R_\sigma - \N' }{\lambda - 1}.
\end{aligned}\end{equation}
This is clearly a valid superchannel, as it can be realised by preparing a maximally entangled state, applying $\K$ to one half of it, performing the measurement $\{ \Phi_+, \id - \Phi_+ \}$, and finally preparing a channel depending on the outcome. It is also easily checked that it achieves the desired transformation, as
\begin{equation}\begin{aligned}
    \Theta[ \idc_K ] = \N'.
\end{aligned}\end{equation}
Now, for this protocol to be replacer preserving, we need that
\begin{equation}\begin{aligned}
    \Theta [ \R_\omega ] = \frac{1}{|K|^2} \, \N' + \frac{|K|^2-1}{|K|^2} \frac{ \lambda \R_\sigma - \N' }{\lambda - 1} \; \in \Rep \;\; \forall \R_\omega \in \Rep.
\end{aligned}\end{equation}
This is achieved for $|K|^2 = \lambda$, which gives the achievability part of the lemma.
\end{proof}

To obtain the asymptotic simulation cost, we then need to evaluate the asymptotics of the quantity $\Dmax{mdia}{n}$. This is made easier by the following lemma, which relates $\Dmax{mdia}{n}$ to a quantity that no longer involves a smoothing over channels, but only over states.

\begin{lemma}\label{lem:dmax_imax}
Define
\begin{equation}\begin{aligned}
    \wh{I}^\ve_{\max} (\N) &\coloneqq \sup_{\psi_{AA}} \inf_{\R_\sigma \in \Rep} \wh{D}^\ve_{\max} (\idc_A \otimes \N(\psi_{AA}) \| \idc_A \otimes \R_\sigma (\psi_{AA}))\\
    &\hphantom{:}= \sup_{\psi_{AA}} \inf_{\sigma_B} \wh{D}^\ve_{\max} (\idc_A \otimes \N(\psi_{AA}) \| \psi_A \otimes \sigma_B).
\end{aligned}\end{equation}
Then
\begin{equation}\begin{aligned}
    \wh{I}^\ve_{\max} (\N) \leq \Dmax{mdia}{n} (\N \| \Rep) \leq \wh{I}^\ve_{\max} (\N )  + \log \frac{1}{1-\ve}.
\end{aligned}\end{equation}
\end{lemma}
\begin{proof}
Using the variational form of $\wh{D}^\ve_{\max}$ in Proposition~\ref{prop:dmax_max_ext}, we can write
\begin{align}\label{eq:imax_sdp}
    2^{\wh{I}^\ve_{\max} (\N)} &= \sup_{\psi_{AA}} \inf \lset \lambda \bar \idc_A \otimes \N(\psi_{AA}) \leq \lambda \psi_R \otimes \sigma_B + Z_{AB}, \right. \nonumber\\
    &\qquad\qquad\qquad\;\;  0 \leq Z_{AB} \leq \idc_A \otimes \N(\psi_{AA}),\; \Tr Z_{AB} \leq \ve \big\} \nonumber\\
    &\texteq{(i)} \sup_{\rho_A > 0} \inf \lset \lambda \bar \left( \sqrt{\rho_A} \otimes \id_B \right) J_{\N}  \left( \sqrt{\rho_A} \otimes \id_B \right) \leq \lambda \rho_A \otimes \sigma_B +  Z_{AB},\right.\nonumber\\
    &\qquad\qquad\qquad\;\; \left. 0 \leq Z_{AB} \leq  \left( \sqrt{\rho_A} \otimes \id_B \right) J_{\N}  \left( \sqrt{\rho_A} \otimes \id_B \right),\; \Tr Z_{AB} \leq \ve  \rset\nonumber\\
    &\texteq{(ii)} \sup_{\rho_A > 0} \inf \lset \lambda \bar  J_{\N}  \leq \lambda \id_A \otimes \sigma_B +   Z'_{AB}, \; 0 \leq Z'_{AB} \leq  J_{\N} ,\; \Tr (\rho_A \otimes \id_B) Z'_{AB} \leq \ve \rset\\
    %&\qquad\qquad\qquad\;\;   0 \leq Z'_{AB} \leq  J_{\N} ,\; \Tr (\rho_A \otimes \id_B) Z'_{AB} \leq \ve \big\}\\
    &\texteq{(iii)} \inf \lsetr \lambda \barr  J_{\N}  \leq \lambda \id_A \otimes \sigma_B +   Z'_{AB}, \;  0 \leq Z'_{AB} \leq  J_{\N},\; \sup_{\rho_A > 0} \Tr \rho_A  Z'_{A} \leq \ve \rsetr \nonumber\\
    &= \inf \lset \lambda \bar  J_{\N}  \leq \lambda \id_A \otimes \sigma_B +   Z'_{AB}, \;  0 \leq Z'_{AB} \leq  J_{\N} ,\; \norm{Z'_{A}}{\infty} \leq \ve \rset, \nonumber
\end{align}
where: in (i) we used the fact that, up to an inconsequential local unitary, any pure state $\psi_{AA}$ can be written as $|A| \left( \sqrt{\rho_A} \otimes \id_B \right) \Phi_A^+ \left( \sqrt{\rho_A} \otimes \id_B \right)$ for some state $\rho_A \geq 0$, and furthermore we restricted ourselves without loss of generality to the interior of the set of states, i.e.\ to $\rho_A > 0$; in (ii), we changed the optimisation variable to $Z'_{AB} \coloneqq \left( \rho_A^{-1/2} \otimes \id_B \right) Z_{AB} \left( \rho_A^{-1/2} \otimes \id_B \right)$ and used the fact that conjugation by a positive definite operator preserves positivity; finally, in (iii) we used Sion's minimax theorem.

Consider now that, by definition,
\begin{equation}\begin{aligned}\label{eq:dmax_mdia_sdp}
    2^{\Dmax{mdia}{n} (\N \| \Rep)} &= \inf \lset \lambda \bar \N' \leq \lambda \R_\sigma, \; \wh{\Diamond}(\N, \N') \leq \ve,\; \R_\sigma \in \Rep \rset\\
    &= \inf \lset \lambda \bar J_{\N'} \leq \lambda \id_A \otimes \sigma_B,\; J_{\N} - J_{\N'} \leq Z_{AB} \leq J_{\N},\; Z_{AB} \geq 0,\; \norm{Z_A}{\infty} \leq \ve\rset
\end{aligned}\end{equation}
where we used the form of the maximal diamond distance in Lemma~\ref{lem:diamond_sdp}. Then clearly, any feasible solution $Z_{AB}$ for ~\eqref{eq:dmax_mdia_sdp} is also feasible for~\eqref{eq:imax_sdp}, implying that $\wh{I}^\ve_{\max} (\N) \leq \Dmax{mdia}{n} (\N \| \Rep)$. For the other inequality, let $Z'_{AB}$ and $\sigma_B$ be feasible solutions for~\eqref{eq:imax_sdp}. Then the operator
\begin{equation}\begin{aligned}
    J_{\M} \coloneqq J_{\N} - Z'_{AB} + Z'_{A} \otimes \sigma_B
\end{aligned}\end{equation}
is positive semidefinite since $Z'_{AB} \leq J_{\N}$, and clearly satisfies $\Tr_B J_{\M} = \id$, which means that it is a valid Choi operator of a channel $\M$. Furthermore, since
\begin{equation}\begin{aligned}
    J_\N - J_\M = Z'_{AB} - Z'_A \otimes \sigma \leq Z'_{AB} \leq J_\N,
\end{aligned}\end{equation}
we have that $\wh{\Diamond}(\N, \M) \leq \ve$. 
Combining the inequalities $J_{\N} - Z'_{AB} \leq \lambda \id_A \otimes \sigma_B$  and $Z'_A \leq \ve \id_A$ we then get
\begin{equation}\begin{aligned}
    J_{\M} \leq J_\N  - Z'_{AB} + Z'_A \otimes \sigma \leq (\lambda + \ve) \id_A \otimes \sigma_B.
\end{aligned}\end{equation}
Overall this means that
\begin{equation}\begin{aligned}
    \Dmax{mdia}{n} (\N \| \Rep)  & \leq D_{\max}(\M \| \R_\sigma)\\
    &\leq \log (\lambda+\ve)\\
    &\leq \log \lambda + \log \frac{1}{1-\ve},
\end{aligned}\end{equation}
where the last line follows since taking the partial trace of the inequality $J_\N \leq \lambda \id_A \otimes \sigma_B + Z'_{AB}$ implies that $\id_A \leq \lambda \id_A + Z'_A \leq \lambda \id_A + \ve \id_A$ and hence that $\lambda \geq 1-\ve$, which then gives $\lambda/(1-\ve) - (\lambda + \ve) \geq 0$. Optimising over all choices of $\lambda$ gives the statement of the lemma.
\end{proof}

As a corollary, we then immediately have that
\begin{equation}\begin{aligned}\label{eq:simulation_cost_as_regularised}
    \wh{S}_{\rm NS}(\N) = \frac12 \lim_{\ve\to0} \limsup_{n\to\infty} \frac1n \Dmax{mdia}{n} (\N^{\otimes n} \| \Rep) = \frac12\lim_{\ve\to0} \limsup_{n\to\infty} \frac1n \wh{I}^\ve_{\max} (\N^{\otimes n}).
\end{aligned}\end{equation}

Since the conventional simulation cost is given by $S_{\rm NS}(\N) = \frac12 I(\N)$, it is tempting to conjecture that $\wh{S}_{\rm NS}(\N) \stackrel{?}{=} \frac12 \wh{I}(\N)$,
where
\begin{equation}\begin{aligned}
    \wh{I}(\N) \coloneqq \inf_{\sigma_B} \sup_{\psi_{AA}} \wh{D}(\idc_A \otimes \N (\psi_{AA}) \| \psi_A \otimes \sigma_B)
\end{aligned}\end{equation}
is the Belavkin--Staszewski mutual information of the channel. 
 However, unlike its Umegaki-based cousin, the Belavkin--Staszewski mutual information $\wh{I}(\N)$ is not always additive --- although numerical searches suggest that violations of additivity are neither common nor significant, in Appendix~\ref{sec:nonadditivity_examples} we discuss a violation of additivity certified with explicit counterexamples. This means that the rates cannot be governed by $\wh{I}(\N)$ in general. However, the additivity of $\wh{I}(\N)$ does hold true for many channels, and indeed for several example of channels we can show that it equals the simulation cost under maximal distance. The implication of this is that a breakdown of the equivalence between quantum capacity and simulation cost can be easily certified for such channels. 

 We formalise a sufficient condition that will allow us to evaluate the exact simulation cost and rule out the reversibility of several types of channels. 
 It is formulated in terms of the R\'enyi mutual information quantities defined using the maximal R\'enyi divergences~\cite{fang_2021}
\begin{equation}\begin{aligned}\label{eq:channel_minimax}
    \wh{I}_\alpha(\N) &\coloneqq \inf_{\sigma_B}  \sup_{\psi_{AA}} \wh{D}_\alpha( \idc_A \otimes \N (\psi_{AA}) \| \psi_A \otimes \sigma_B )\\
    &\hphantom{:}= \sup_{\psi_{AA}} \inf_{\sigma_B} \wh{D}_\alpha( \idc_A \otimes \N (\psi_{AA}) \| \psi_A \otimes \sigma_B ),
\end{aligned}\end{equation}
where the second line follows a standard minimax argument for channel divergences~\cite{wilde_2020,fang_2021}.

 \begin{boxed}
 \begin{lemma}\label{lem:additivity_makes_revshannon_break_down}
If a quantum channel $\N$ satisfies $\wh{I}_\alpha(\N^{\otimes n}) = n \wh{I}_\alpha(\N)$ for all $\alpha \in (0,1)$ and all $n \in \NN$, then
\begin{equation}\begin{aligned}
    \wh{S}_{\rm NS}(\N) = \frac12 \wh{I}(\N).
\end{aligned}\end{equation}

In particular, if it furthermore holds that $\wh{I}(\N) > I(\N)$, then
\begin{equation}\begin{aligned}
 \wh{S}_{\Omega}(\N) \geq \frac12 \wh{I}(\N) > \frac12 I(\N) \geq \wh{Q}_{\Omega}(\N)
\end{aligned}\end{equation}
 for any type of assistance $\Omega \subseteq \rm NS$, and hence for this channel the reverse Shannon theorem breaks down under the maximal distance.
\end{lemma}
\end{boxed}
\begin{remark}
We stress here that our result for simulation applies to no-signalling assistance, which leads to the strongest irreversibility statement: if $\wh{S}_{\rm NS}(\N) > \wh{Q}_{\rm NS}(\N)$, then it also holds that $\wh{S}_{\Omega}(\N) > \wh{Q}_{\Omega}(\N)$ under any type of assistance $\Omega \subseteq {\rm NS}$, including shared entanglement.
\end{remark}
\begin{proof}
Proposition~\ref{prop:dmax_aep} immediately gives us the bounds
\begin{equation}\begin{aligned}\label{eq:Imax_bounds}
   \wh{I}_\alpha (\N) - \frac{1}{1-\alpha} \log \frac{1}{1-\ve} \leq \wh{I}^\ve_{\max} (\N) \leq \wh{I}_\beta(\N)  + \frac{1}{\beta-1} \log \frac{1}{\ve} - \log\frac{1}{1-\ve},
\end{aligned}\end{equation}
which hold for all $\ve \in (0,1)$, all $\alpha \in (0,1)$, and all $\beta \in (1,2]$.
Combining Lemma~\ref{lem:oneshot_cost_ns} with Lemma~\ref{lem:dmax_imax} lets us bound $\wh{S}^\ve_{\rm NS}(\N)$ with $\wh{I}^\ve_{\max}$, and~\eqref{eq:Imax_bounds} then tells us that
\begin{equation}\begin{aligned}\label{eq:simcost_twosided_bounds}
   \wh{I}_\alpha (\N) - \frac{1}{1-\alpha} \log \frac{1}{1-\ve}  \leq 2 \wh{S}^\ve_{\rm NS}(\N) \leq \wh{I}_\beta(\N)  + \frac{1}{\beta-1} \log \frac{1}{\ve}.
\end{aligned}\end{equation}

Now, subadditivity of the maximal R\'enyi mutual information is straightforward: by definition,
\begin{equation}\begin{aligned}
    \wh{I}_\beta(\N^{\otimes n}) &\leq \inf_{\sigma_B}  \sup_{\psi_{A^nA^n}} \wh{D}_\beta ( \idc_A \otimes \N^{\otimes n} (\psi_{A^nA^n}) \| \psi_{A^n} \otimes \sigma_B^{\otimes n} )\\
    &= \inf_{\sigma_B}  \sup_{\psi_{A A}} n \wh{D}_\beta ( \idc_A \otimes \N (\psi_{AA}) \| \idc_A \otimes \R_{\sigma_B} (\psi_{AA}) )\\
    &= n \wh{I}_\beta(\N)
\end{aligned}\end{equation}
where the second line follows by the additivity of the geometric R\'enyi divergence for channels for all $\beta \in (1,2]$~\cite[Theorem~3]{fang_2021}. Plugging this into~\eqref{eq:simulation_cost_as_regularised} and taking the limit $\beta \to 1^+$ gives
\begin{equation}\begin{aligned}
    \wh{S}_{\rm NS}(\N) \leq \frac12 \wh{I}(\N)
\end{aligned}\end{equation}
for all channels $\N$.

For the opposite inequality, we need to use the assumption of the lemma. Namely, plugging the lower bound of~\eqref{eq:simcost_twosided_bounds} into~\eqref{eq:simulation_cost_as_regularised}, the property that $\wh{I}_\alpha(\N^{\otimes n}) = n \wh{I}_\alpha(\N)$ allows us to take the limit $\alpha \to 1^-$ and obtain $\wh{S}_{\rm NS}(\N) \geq \frac12 \wh{I}(\N)$. 
\end{proof}

%%%%%%%%%%%%%%%%%%%%%%%%%%%%%%%%%%%%%%%%%%%%%%%%%%%%%%%%%%%%%%%%%%%%%%%%%%%%%%%%%%%

\subsection{Examples of irreversible channels}\label{sec:additivity_examples}

The additivity required for Lemma~\ref{lem:additivity_makes_revshannon_break_down} to hold can be verified directly for the following two classes of channels. 

\begin{boxed}
\begin{proposition}[(Pure-state output channels)]\label{prop:pure_cq_additivity}
Let $\N(\rho) = \sum_x \braket{x|\rho|x} \proj{\phi_x}$ be a classical--quantum channel with pure output states, and define
\begin{equation}\begin{aligned}
    C_\N \coloneqq \min \lset \Tr T \bar T \geq \proj{\phi_x}\ \forall x \rset.
\end{aligned}\end{equation}
Then, for every $n \geq 1$ and $\alpha \in (0,2]$,
\begin{equation}\begin{aligned}
    \wh{I}_\alpha(\N^{\otimes n}) = n \log C_\N,
    \qquad \wh{S}_{\rm NS}(\N) = \frac12 \wh{I}(\N) = \frac12 \log C_\N.
\end{aligned}\end{equation}
\end{proposition}
\end{boxed}
\begin{proof}
For a pure state, $\wh{D}_\alpha(\proj{\phi}\|\sigma) = D_{\max}(\proj{\phi}\|\sigma)$. Thus basis inputs give $\wh{I}_\alpha(\N) \geq \log C_\N$. Conversely, an optimal $T=C_\N\sigma$ satisfies $\N \leq C_\N\R_\sigma$, and $\wh{D}_\alpha \leq D_{\max}$ gives the reverse inequality. The SDP dual is
\begin{equation}\begin{aligned}
    C_\N = \max \lset \sum_x \Tr M_x \proj{\phi_x} \bar M_x \geq 0,\; \sum_x M_x = \id \rset.
\end{aligned}\end{equation}
Tensor products of optimal primal and dual solutions give $C_{\N^{\otimes n}}=C_\N^n$, proving additivity. The simulation rate follows from~\eqref{eq:Imax_bounds} and Lemmas~\ref{lem:oneshot_cost_ns} and~\ref{lem:dmax_imax}.

For a concrete irreversible example, take $\ket{\phi_\pm}=(3\ket{0}\pm\ket{1})/\sqrt{10}$. Binary state discrimination gives $C_\N=1+\sqrt{1-|\braket{\phi_+|\phi_-}|^2}=8/5$, while the Holevo information is maximised by the uniform input. Hence
\begin{equation}\begin{aligned}
    \wh{I}(\N)=\log\frac85 > h\!\left(\frac1{10}\right)=I(\N),
\end{aligned}\end{equation}
where $h$ denotes the binary entropy.
\end{proof}

\begin{boxed}
\begin{proposition}[(Amplitude-damping channels)]\label{prop:amplitude_damping_additivity}
Let $\mathcal A_\gamma$ be the qubit channel with Kraus operators
\begin{equation}\begin{aligned}
    K_0=\proj{0}+\sqrt{1-\gamma}\,\proj{1},
    \qquad K_1=\sqrt{\gamma}\,\ketbraa{0}{1},
    \qquad 0<\gamma<1.
\end{aligned}\end{equation}
For every $n\geq1$ and $\alpha\in(0,2]$,
\begin{equation}\begin{aligned}
    \wh{I}_\alpha(\mathcal A_\gamma^{\otimes n})
    &=2n\log\left(1+\sqrt{1-\gamma}\right),\\
    \wh{S}_{\rm NS}(\mathcal A_\gamma)
    &=\frac12\wh{I}(\mathcal A_\gamma)
    =\log\left(1+\sqrt{1-\gamma}\right).
\end{aligned}\end{equation}
\end{proposition}
\end{boxed}
\begin{proof}
Put $s=\sqrt{1-\gamma}$ and
$\sigma_*=(\proj{0}+s\proj{1})/(1+s)$. Direct calculation gives
\begin{equation}\begin{aligned}
    J_{\mathcal A_\gamma}
    \leq(1+s)^2\id_A\otimes\sigma_*.
\end{aligned}\end{equation}
Taking tensor powers and using $\wh{D}_\alpha\leq D_{\max}$
proves the upper bound.

For the lower bound, write $J=J_{\mathcal A_\gamma^{\otimes n}}$
and define the unnormalised vectors
\begin{equation}\begin{aligned}
    \ket{\Omega}&=\sum_{x\in\{0,1\}^n}\ket{xx},
    \qquad
    \ket{v}=\sum_{x\in\{0,1\}^n}s^{|x|}\ket{xx},
    \qquad c_n=(1+s)^{2n},
\end{aligned}\end{equation}
where $|x|$ is the Hamming weight. The Kraus operator
$K_0^{\otimes n}$ contributes the rank-one term $\proj{v}$ to $J$.
Every other Kraus operator has zero diagonal and annihilates
$\ket{0^n}$. Consequently, every positive operator $X$ supported
on $\mathrm{supp}\,J$ satisfies
\begin{equation}\begin{aligned}
\label{eq:amplitude_support_identity}
    \braket{0^n|\Tr_{B^n}X|0^n}
    =\frac1{c_n}\braket{\Omega|X|\Omega}.
\end{aligned}\end{equation}

Fix $0<\alpha<1$ and an arbitrary positive definite state
$\sigma_{B^n}$. Set $\Sigma=\id_{A^n}\otimes\sigma_{B^n}$ and
\begin{equation}\begin{aligned}
    X_\alpha=\Sigma^{1/2}
    \left(\Sigma^{-1/2}J\Sigma^{-1/2}\right)^\alpha
    \Sigma^{1/2}.
\end{aligned}\end{equation}
Since $\braket{\Omega|\Sigma|\Omega}=1$, Jensen's inequality
for the concave function $u\mapsto u^\alpha$ gives
\begin{equation}\begin{aligned}
    \braket{\Omega|X_\alpha|\Omega}
    \leq\braket{\Omega|J|\Omega}^\alpha=c_n^\alpha.
\end{aligned}\end{equation}
Moreover, $\mathrm{supp}\,X_\alpha=\mathrm{supp}\,J$,
so~\eqref{eq:amplitude_support_identity} implies
\begin{equation}\begin{aligned}
    \lambda_{\min}\!\left(\Tr_{B^n}X_\alpha\right)
    &\leq\braket{0^n|\Tr_{B^n}X_\alpha|0^n}\\
    &\leq c_n^{\alpha-1}.
\end{aligned}\end{equation}
For a full-rank reference marginal $\rho$, the geometric
R\'enyi trace functional is $\Tr[(\rho\otimes\id)X_\alpha]$.
Optimising over such marginals, and recalling that
$\alpha-1<0$, therefore gives
\begin{equation}\begin{aligned}
    &\sup_{\psi}\wh{D}_\alpha\!\left(
    (\idc\otimes\mathcal A_\gamma^{\otimes n})(\psi)
    \,\middle\|\,\psi_{A^n}\otimes\sigma_{B^n}\right)\\
    &\qquad\geq\frac1{\alpha-1}
    \log\lambda_{\min}\!\left(\Tr_{B^n}X_\alpha\right)
    \geq\log c_n.
\end{aligned}\end{equation}
As $\sigma_{B^n}$ was arbitrary, this proves the lower bound;
positive definite states suffice by regularisation.
Monotonicity in $\alpha$ extends the lower bound to
$1\leq\alpha\leq2$. The simulation rate follows as above.

At $\gamma=1/2$, the channel equals its complementary channel,
so the mutual information for an input purification is simply
the input entropy. Consequently,
\begin{equation}\begin{aligned}
    \wh{S}_{\rm NS}(\mathcal A_{1/2})
    &=\log\left(1+\frac1{\sqrt2}\right) \\
    &>\frac12\log2
    =\frac12 I(\mathcal A_{1/2})
    =\wh{Q}_{\rm NS}(\mathcal A_{1/2}),
\end{aligned}\end{equation}
certifying the irreversibility of this channel.
\end{proof}

%%%%%%%%%%%%%%%%%%%%%%%%%%%%%%%%%%%%%%%%%%%%%%%%%%%%%%%%%%%%%%%%%%%%%%%%%%%%%%%%%%%
%%%%%%%%%%%%%%%%%%%%%%%%%%%%%%%%%%%%%%%%%%%%%%%%%%%%%%%%%%%%%%%%%%%%%%%%%%%%%%%%%%%
%%%%%%%%%%%%%%%%%%%%%%%%%%%%%%%%%%%%%%%%%%%%%%%%%%%%%%%%%%%%%%%%%%%%%%%%%%%%%%%%%%%

\section{Discussion}

We have shown that the consequences of departing from conventional, minimal extensions of distances used in the study of quantum information heavily depend on the task under consideration: for distillation-type problems under mild assumptions, all extensions of trace or purified distance are equivalent, while for dilution-type problems, the change to the maximal extension causes a complete upheaval of the asymptotic rates.
Our results show in particular how the Belavkin--Staszewski relative entropy $\wh{D}(\rho\|\sigma)$ and entropic quantities built from it can emerge as the asymptotic rates of  quantum information processing tasks, thereby preventing a quantification of all asymptotic rates in terms of the same entropic measures and causing a breakdown of reversibility phenomena in quantum resources and quantum channel manipulation. 

Our investigation was inspired by, and in many parts directly relied on, the characterisation of maximal extensions of quantum divergences by Matsumoto~\cite{matsumoto_2010,matsumoto_2010-1,matsumoto_2018}. This conceptually interesting formalisation has been impactful in the understanding of the mathematical foundations of quantum information measures, but its connection with operational problems has remained murky --- we hope to have shed some light on this here. 
We stress that the shift from minimal distance measures to larger ones imposes \emph{more} demanding requirements on the errors of quantum information processing, and does not require operational relaxations or allowances. While the definitions of Matsumoto's maximal extensions employ preparation channels that map classical distributions to different quantum states, we do not require access to such structures in the tasks that we studied here. Indeed, a small maximal distance between two quantum states ensures that they remain hard to distinguish even if their preparation schemes are exposed, constituting the most conservative error criterion consistent with the classical ones.

The discrepancy between the two types of problems that we studied can be understood as coming from the structure of the output states of these protocol: distillation requires the output states to closely approximate a pure state, allowing for the error to be evaluated on the simpler class of isotropic states, while for dilution the target states are completely generic, and imposing the more restrictive error criteria based on maximal distances prevents the standard asymptotic rates from being achievable.
The structural difference between distillation and dilution problems is very closely related to the division of information-theoretic problems into packing and covering type tasks, and we expect that similar distinctions can be obtained more broadly, with quantum packing-type problems governed by entropic measures based on the Umegaki relative entropy even under maximal distances, while covering-type problems undergoing a change to the Belavkin--Staszewski relative entropy. 
The gap here emerges already at the level of first-order asymptotic rates, and we have not studied higher-order or large-deviation refinements to the asymptotic analysis, which would be an interesting point for further investigation. In this direction, our Proposition~\ref{prop:strong_converse_exponent} suggests that maximal R\'enyi divergences could find use in the analysis of exponents of covering tasks under maximal distances.

%%%%%%%%%%%%%%%%%%%%%%%%%%%%%%%%%%%%%%%%%%%%%%%%%%%%%%%%%%%%%%%%%%%%%%%%%%%%%%%%%%%
%%%%%%%%%%%%%%%%%%%%%%%%%%%%%%%%%%%%%%%%%%%%%%%%%%%%%%%%%%%%%%%%%%%%%%%%%%%%%%%%%%%
%%%%%%%%%%%%%%%%%%%%%%%%%%%%%%%%%%%%%%%%%%%%%%%%%%%%%%%%%%%%%%%%%%%%%%%%%%%%%%%%%%%

\begin{acknowledgments}

B.R.\ acknowledges the support of the Japan Science and Technology Agency (JST) PRESTO grant no.\ JPMJPR25FB and the RIKEN Pioneering Project
`Mathematical Foundation of Quantum Information'. L.L.\ acknowledges financial support from the European Union (ERC StG ETQO, Grant Agreement no.\ 101165230).  M.B.\ acknowledges support from the European Research Council (ERC Grant Agreement No.~948139) and the Excellence Cluster Matter and Light for Quantum Computing (ML4Q-2). M.T.\ acknowledges support from the National Research Foundation Investigatorship Award (NRF-NRFI10-2024-0006) and the National Research Foundation, Singapore, through the National Quantum Office, hosted by A*STAR, under its Centre for Quantum Technologies Funding Initiative (S24Q2d0009).

\medskip
\textbf{Statement on AI use.}\,---\, The examples of channels in Section~\ref{sec:additivity_examples} and Appendix~\ref{sec:nonadditivity_examples} were obtained by consulting the omniscient oracle in the form of ChatGPT 6 Astra. All other content in the paper, including the main ideas and technical proofs, was organically sourced from free-range humans and is 100\% natural.
\end{acknowledgments}

%%%%%%%%%%%%%%%%%%%%%%%%%%%%%%%%%%%%%%%%%%%%%%%%%%%%%%%%%%%%%%%%%%%%%%%%%%%%%%%%%%%

 \bibliographystyle{apsca}
 \bibliography{main}

 \appendix

 \section{Resource cost in more general settings}\label{app:general}

We now clarify exactly what assumptions need to be made about the given resource theory and the state $\phi_\star$ in order for the statements of Section~\ref{sec:resources} to be applicable in more general quantum resource theoretic contexts, similar to the assumptions found e.g.\ in~\cite{liu_2019}.

First, we will need to restrict ourselves to the class of resource--non-generating (RNG) operations, defined as the channels $\N$ such that $\sigma \in \FF \Rightarrow \N(\sigma) \in \FF$. This is the largest consistent class of free operations, since the condition of resource non-generation is a natural requirement that any other type of free channels should obey; the resource cost under RNG thus lower bounds the resource cost under any other free operations. This is of particular relevance here since showing gaps between resource cost and distillable resource under RNG rules out reversibility under all smaller classes of operations as well.

Furthermore, we note that choosing $\phi_\star$ as a maximally resourceful state, as is typically done in the context of resource dilution, already imposes constraints on its properties. In particular, \cite[Theorem~4]{regula_2020} shows that if $\phi_\star$ maximises the resource measure $D_{\max}(\cdot\|\FF)$ among all states, then it also maximises $D_{\min}(\cdot\|\FF)$ where $D_{\min}(\rho\|\sigma) = \overline{D}_0(\rho\|\sigma)$, and in fact $D_{\max}(\phi_\star\|\FF) = D_{\min}(\phi_\star\|\FF)$. This means that all R\'enyi-based resource measures collapse to a single value for the maximally resourceful state. By extending this property to multiple copies, we use it to formulate the following statement.

\begin{proposition}
Let $\FF$ be a convex and closed set of free states of the given resource theory, and let 
Let $\phi_\star$ be the pure resource unit state. Assume that, for all $k \in \NN$, $D_{\max}(\phi_\star^{\otimes k}\|\FF) = D_{\min}(\phi_\star^{\otimes k}\|\FF) = k \log c$ for some $c > 1$. 

\medskip
\noindent \textbf{(a)} If copies of $\phi_\star$ have a constant overlap with the free states --- that is, $\Tr\phi_\star^{\otimes k} \sigma_k = c^{-k}$  for all $\sigma_k \in \FF$ --- then
\begin{equation}\begin{aligned}
    \mathrm{Cost}_{\rm RNG}^{\ve,\wh{T}}(\rho) = \ceil{ \frac{\wh{D}_{\max}^{\ve,=}(\rho \| \FF)}{\log c} }
\end{aligned}\end{equation}
for all states $\rho$.

\medskip
\noindent \textbf{(b)} Alternatively, if the target state $\rho = \psi$ is pure and satisfies that $D_{\max}(\psi^{\otimes n} \| \FF) = D_{s,\FF}(\psi^{\otimes n}\|\FF)$ for all $n \in \NN$, where
\begin{equation}\begin{aligned}
    D_{s,\FF}(\psi\|\FF) \coloneqq \log \inf \lset \mu \bar \psi + (\mu - 1) \sigma_- = \mu \sigma_+,\; \sigma_\pm \in \FF \rset
\end{aligned}\end{equation}
denotes the (logarithmic) standard robustness, then also
\begin{equation}\begin{aligned}
    \mathrm{Cost}_{\rm RNG}^{\ve,\wh{T}}(\psi) = \ceil{ \frac{\wh{D}_{\max}^{\ve,=}(\psi \| \FF)}{\log c} }.
\end{aligned}\end{equation}

\medskip 
In either of the two cases, we then have that
\begin{equation}\begin{aligned}
    R_{\wh{T}}(\phi_\star \xrightarrow[\rm RNG]{} \rho)^{-1} = \frac{\limsup_{\ve\to0}\limsup_{n\to\infty}\frac1n \wh{D}_{\max}^\ve(\rho^{\otimes n} \| \FF)}{\log c}.
\end{aligned}\end{equation}
\end{proposition}

The proof is completely analogous to the ones in  Section~\ref{sec:resources_cost}.

One commonly studied example of a quantum resource theory that satisfies condition (a) is the theory of coherence~\cite{baumgratz_2014}, where the free states are all diagonal states in a fixed basis. Condition (b) is satisfied for example in the more general resource theory of $k$-coherence~\cite{johnston_2018}, as well the related theory of $k$-entanglement (entanglement of bounded Schmidt rank $k$).

\section{Counterexamples to additivity of maximal mutual information}
\label{sec:nonadditivity_examples}

\begin{proposition}[(Mixed-output classical--quantum channels)]
\label{prop:mixed_cq_nonadditivity}
Let
\begin{equation}\begin{aligned}
    \rho_\pm=\begin{pmatrix}a&\pm c\\ \pm c&1-a\end{pmatrix},
    \qquad 0<c^2<a(1-a),\quad a\ne\frac12,
\end{aligned}\end{equation}
and define
\begin{equation}\begin{aligned}
    \mathcal E(\omega)
    =\braket{0|\omega|0}\rho_+
     +\braket{1|\omega|1}\rho_-.
\end{aligned}\end{equation}
Then
\begin{equation}\begin{aligned}
    \wh I(\mathcal E^{\otimes2})<2\wh I(\mathcal E).
\end{aligned}\end{equation}
The same strict inequality holds for $\wh I_\alpha$ whenever
$\alpha$ is sufficiently close to $1$.
\end{proposition}
\begin{proof}
Put $\rho=\rho_+$ and $Z=\operatorname{diag}(1,-1)$.
For a classical--quantum channel, the divergence to a replacer
is the maximum divergence of its output states. Since
$\rho_-=Z\rho Z$, convexity and averaging a comparison state
with its conjugate by $Z$ give
\begin{equation}\begin{aligned}
    \wh I(\mathcal E)
    =\min_{0<s<1}\wh D(\rho\|\sigma_s),
    \qquad \sigma_s=\operatorname{diag}(s,1-s).
\end{aligned}\end{equation}
Choose a minimizer $s$ and write $\sigma=\sigma_s$.
It is interior because $\rho>0$ and the divergence tends to
infinity at both endpoints. Define
\begin{equation}\begin{aligned}
    Y=\rho^{-1/2}\sigma\rho^{-1/2},
    \qquad H=\rho^{-1/2}Z\rho^{-1/2}.
\end{aligned}\end{equation}
The real-valued function
$f(s)=(\ln2)\wh D(\rho\|\sigma_s)
=-\Tr\rho\ln(\rho^{-1/2}\sigma_s\rho^{-1/2})$
satisfies $f'(1/2)=-2(2a-1)\ne0$, so $s\ne1/2$.

Take a real orthonormal eigenbasis of $Y$, with eigenvalues
$x,y$, and express all matrix entries below in this basis.
Set
\begin{equation}\begin{aligned}
    L(u,v)&=\frac{\ln u-\ln v}{u-v},\qquad L(u,u)=\frac1u,\\
    C&=\rho_{12}H_{21},\qquad
    m=\frac{\rho_{11}H_{11}}x+\frac{\rho_{22}H_{22}}y.
\end{aligned}\end{equation}
The derivative of the matrix logarithm has entries
$[\mathrm d\ln(Y)[H]]_{ij}=L(x_i,x_j)H_{ij}$.
Thus the stationarity condition $f'(s)=0$ reads
\begin{equation}\begin{aligned}
    m=-2C L(x,y).
\end{aligned}\end{equation}

Now introduce a correlation between the comparison states:
\begin{equation}\begin{aligned}
    \sigma_\varepsilon
    =\sigma\otimes\sigma+\varepsilon Z\otimes Z,
    \qquad
    g(\varepsilon)
    =(\ln2)\wh D(\rho^{\otimes2}\|\sigma_\varepsilon).
\end{aligned}\end{equation}
For sufficiently small $\varepsilon>0$, this is a positive
definite state with the same marginals as $\sigma\otimes\sigma$.
Applying the same logarithmic derivative formula to
$Y\otimes Y+\varepsilon H\otimes H$ gives
\begin{equation}\begin{aligned}
    g'(0)
    &=-m^2-4C L(x,y)m
      -2C^2\left[L(x^2,y^2)+\frac1{xy}\right]\\
    &=-2C^2\left[L(x^2,y^2)+\frac1{xy}-2L(x,y)^2\right].
\end{aligned}\end{equation}
Here $x\ne y$, since otherwise $\sigma$ would be proportional
to the nondiagonal $\rho$. Also $C\ne0$: the relation
$[\rho,\sigma]\ne0$ implies $\rho_{12}\ne0$ in the chosen
basis, and
$Y=\rho^{-1}/2+(s-1/2)H$ gives
$H_{12}=\rho_{12}/[2(s-1/2)\det\rho]\ne0$.
Finally, writing $x=k e^t$ and $y=k e^{-t}$, the bracket equals
\begin{equation}\begin{aligned}
    \frac{\sinh^2t+t\tanh t-2t^2}{k^2\sinh^2t}>0.
\end{aligned}\end{equation}
Indeed, the numerator is even, and for $t>0$ use
$\sinh^2t>t^2+t^4/3$ and $\tanh t\geq t-t^3/3$.
Hence $g'(0)<0$.

The diagonal state $\sigma_\varepsilon$ has the same divergence
from all four outputs of $\mathcal E^{\otimes2}$. Therefore,
for sufficiently small $\varepsilon>0$,
\begin{equation}\begin{aligned}
    \wh I(\mathcal E^{\otimes2})
    &\leq\wh D(\rho^{\otimes2}\|\sigma_\varepsilon)\\
    &<\wh D(\rho^{\otimes2}\|\sigma^{\otimes2})
     =2\wh I(\mathcal E).
\end{aligned}\end{equation}
Since $\rho_\pm$ are positive definite, the optimized geometric
R\'enyi quantities are continuous at $\alpha=1$: convergence
from below follows by compactness of the comparison states and
monotonicity in $\alpha$, while convergence from above follows
by evaluating at a positive definite BS minimizer.
The same argument applies to two copies. Hence the strict
inequality persists for orders sufficiently close to $1$.
\end{proof}

For example, the proposition applies to
$\rho_\pm=\left(\begin{smallmatrix}4/5&\pm1/5\\
\pm1/5&1/5\end{smallmatrix}\right)$.
It also applies to every noisy version
\begin{equation}\begin{aligned}
    \rho_\pm=(1-\eta)\proj{\phi_\pm}+\eta\frac{\id}{2},
    \qquad
    \ket{\phi_\pm}=\frac{3\ket0\pm\ket1}{\sqrt{10}},
    \qquad 0<\eta<1.
\end{aligned}\end{equation}
Thus arbitrarily weak depolarizing noise destroys the additivity
of the corresponding pure-output example.

\begin{corollary}[(Quantum channels with noisy flags)]
\label{cor:quantum_nonadditivity}
With $\rho_\pm$ and $\mathcal E$ as in
Proposition~\ref{prop:mixed_cq_nonadditivity}, define
\begin{equation}\begin{aligned}
    \mathcal F(\omega)
    =\frac12\omega\otimes\rho_+
     +\frac12 Z\omega Z\otimes\rho_-.
\end{aligned}\end{equation}
Then $\mathcal F$ is not entanglement breaking and
\begin{equation}\begin{aligned}
    \wh I(\mathcal F^{\otimes n})
    =n+\wh I(\mathcal E^{\otimes n}),\qquad n\geq1.
\end{aligned}\end{equation}
In particular, $\wh I(\mathcal F^{\otimes2})<2\wh I(\mathcal F)$.
\end{corollary}
\begin{proof}
The identity
$\mathcal F(P\omega P^\dagger)
=(P\otimes\id)\mathcal F(\omega)(P^\dagger\otimes\id)$
for every Pauli operator $P$ permits averaging the comparison
state over independent Pauli operators on the transmitted qubits.
By convexity, we may therefore restrict it to
$\id_{B^n}/d\otimes\tau$, where $d=2^n$ and $\tau$ is an
arbitrary state of the $n$ flag systems.

Write $\rho_x=\bigotimes_{i=1}^n Z^{x_i}\rho_+Z^{x_i}$.
The Choi matrix is
\begin{equation}\begin{aligned}
    J_{\mathcal F^{\otimes n}}
    =\sum_{x\in\{0,1\}^n}\Phi_x\otimes\rho_x,
\end{aligned}\end{equation}
where $\Phi_x$ are mutually orthogonal normalized maximally
entangled projectors, with $\Tr_{B^n}\Phi_x=\id_{R^n}/d$.
Applying the BS Choi formula~\cite[Theorem~3]{fang_2021}
to these orthogonal blocks gives
\begin{equation}\begin{aligned}
    \wh D\!\left(\mathcal F^{\otimes n}
    \middle\|\R_{\id/d\otimes\tau}\right)
    =n+\frac1d\sum_x\wh D(\rho_x\|\tau).
\end{aligned}\end{equation}
The unitary $Z_B\otimes Z_F$ leaves every output of
$\mathcal F$ invariant. Averaging its independent applications
allows $\tau$ to be diagonal. All terms in the preceding sum
then coincide, giving
\begin{equation}\begin{aligned}
    \wh I(\mathcal F^{\otimes n})
    &=n+\min_{\tau\ {\rm diagonal}}
      \wh D(\rho_+^{\otimes n}\|\tau)\\
    &=n+\wh I(\mathcal E^{\otimes n}).
\end{aligned}\end{equation}
Finally, for input $\id/2$, the coherent information is
$h(a)-S(\rho_+)>0$, since dephasing the nondiagonal $\rho_+$
strictly increases its entropy. Hence $\mathcal F$ is not
entanglement breaking.
\end{proof}

For either family, applying the simulation upper bound to
two-copy blocks yields
\begin{equation}\begin{aligned}
    \wh S_{\rm NS}(\N)
    \leq\frac14\wh I(\N^{\otimes2})
    <\frac12\wh I(\N).
\end{aligned}\end{equation}
Thus these channels also disprove the single-letter simulation
conjecture.

\end{document}